\documentclass[letterpaper,11pt,oneside]{article}
\usepackage{graphicx}
\usepackage{enumitem}
\usepackage{amssymb}
\usepackage{amsthm}
\usepackage{amsmath}
\usepackage{tikz}
\usetikzlibrary{shapes}
\usetikzlibrary{patterns}
\usetikzlibrary{arrows,positioning,decorations.pathmorphing,trees}
\usetikzlibrary{arrows.meta}
\usepackage{cancel}
\usepackage{tikz-cd}
\makeatletter
\tikzset{
  column sep/.code=\def\pgfmatrixcolumnsep{\pgf@matrix@xscale*(#1)},
  row sep/.code   =\def\pgfmatrixrowsep{\pgf@matrix@yscale*(#1)},
  matrix xscale/.code=
    \pgfmathsetmacro\pgf@matrix@xscale{\pgf@matrix@xscale*(#1)},
  matrix yscale/.code=
    \pgfmathsetmacro\pgf@matrix@yscale{\pgf@matrix@yscale*(#1)},
  matrix scale/.style={/tikz/matrix xscale={#1},/tikz/matrix yscale={#1}}}
\def\pgf@matrix@xscale{1}
\def\pgf@matrix@yscale{1}
\makeatother
\usepackage{stackengine}
\usepackage{pgfplots}
\usepackage{float}
\usepackage{braket}
\usepackage{setspace}
\usepackage[ruled, lined, linesnumbered, commentsnumbered, longend]{algorithm2e}
\usepackage{fullpage}
\usepackage{comment}
\usepackage[square,numbers]{natbib}
\usepackage{authblk}
\usepackage{hyperref}
\hypersetup{
    colorlinks=true,
    linkcolor=purple,
    citecolor=purple
    }
\usepackage{cleveref}

\newcommand{\bilin}[2]{\left\langle #1 , #2 \right\rangle}
\newcommand{\TC}[2]{\mathcal{TC}_{#1}[#2]}
\newcommand{\NC}[2]{\mathcal{NC}_{#1}[#2]}
\newcommand{\proj}[2]{\Pi_{#1}[#2]}
\newcommand{\poly}{\textnormal{poly}}

\newtheorem{theorem}{Theorem}[section]

\newtheorem{corollary}[theorem]{Corollary}

\newtheorem{lemma}[theorem]{Lemma}

\newtheorem{proposition}[theorem]{Proposition}
\newtheorem{definition}{Definition}

\newtheorem{observation}{Observation}
\newtheorem{example}{Example}
\newtheorem{remark}{Remark}
\newtheorem*{remark*}{Remark}
\newtheorem*{inftheorem*}{Theorem (informal)}
\newtheorem*{notation*}{Notation}
\newtheorem*{observation*}{Observation}
\newtheorem*{theorem*}{Theorem}
\newtheorem*{proposition*}{Proposition}
\newtheorem*{definition*}{Definition}
\newtheorem*{axiom*}{Axiom}
\newtheorem*{claim*}{Claim}
\newtheorem*{lemma*}{Lemma}

\title{Gradient Dynamics in First-Price Auctions: \\ Iterative Strategy Elimination via Cubic Potentials}
\author[1]{Mete \c{S}eref Ahunbay\thanks{E-mail: \texttt{mete.ahunbay@cs.ox.ac.uk}}}
\author[2]{Weiqiang Zheng\thanks{E-mail: \texttt{weiqiang.zheng@yale.edu}}}
\author[3,4]{Tao Lin\thanks{E-mail: \texttt{lintao@cuhk.edu.cn}}}
\affil[1]{Department of Computer Science, University of Oxford}
\affil[2]{Department of Computer Science, Yale University}
\affil[3]{Microsoft Research}
\affil[4]{School of Data Science, The Chinese University of Hong Kong, Shenzhen}
\date{3rd June, 2026}

\begin{document}
\maketitle

\begin{abstract}
We show that in discretised first-price auctions with complete information, if the buyers learn to bid with online gradient ascent, in time-average the outcome is (almost) the efficient outcome of the second-price auction. Our proof rests on two novel innovations in the analysis of online gradient ascent in normal-form games, which may be useful in a wider range of applications. First, we develop a potential-function-based argument for the analysis of gradient ascent in normal-form games, allowing us to deduce that certain strategies will not be played in time-average. We provide sufficient conditions which ensure this argument can be applied iteratively, resulting in a procedure reminiscent of iterative elimination of dominated strategies. Second, we develop a novel class of cubic \emph{``candidate potential functions''}, classifying a family of quadratic strategy modifications on the probability simplex against which online gradient ascent incurs no regret.
\end{abstract}

\allowdisplaybreaks

\section{Introduction}

In a discretised, complete-information first-price auction, $N \ge 2$ buyers with integers values $v_1 \geq v_2 \geq \cdots \geq v_N$ choose integer bids between $0$ and $v_i-1$ to compete for one item.
Even in such a basic setting, the following question regarding learning dynamics is not resolved in the literature: 

\begin{center}
    \emph{When every buyer runs the online gradient ascent algorithm in repeated first-price auctions, \\
    what is the long-term  
    outcome (e.g., welfare and revenue)
    of the resulting dynamics?}
\end{center}

Aside from its obvious theoretical importance as a fundamental setting in economics, understanding the dynamics of learning behaviour in first-price auctions is also of practical relevance. For instance, online advertising markets use first-price auctions to allocate ad slots, transitioning from the classic second-price auctions. 
It has been a standard practice that buyers employ automated bidding algorithms to adaptively adjust bidding strategies based on historical observations \citep{aggarwal2024auto}.
The interaction between these bidding algorithms forms a complex learning dynamic. Moreover, from the platform's perspective, it is important to understand how learning algorithms affect the social welfare and revenue of first-price auctions.

Simple as it may seem, the above question turns out to be highly non-trivial. In the discretised setting, it is known\footnote{\citet{FLN16} prove exact revenue / social welfare equivalence with the second-price auction for the model with continuous bidding sets, however, the arguments provided extend readily to our discretised setting.} that all correlated equilibria (almost) achieve the efficient second-price outcome, i.e., the highest-value bidder wins and pays at least $v_2-2$ \cite{FLN16}.
However, algorithms that are guaranteed to converge to correlated equilibria -- in particular, no-internal or no-swap regret algorithms --
are often too intricate, as they require fixed-point computation \cite{foster1997calibrated,blum2007external} or elaborate constructions averaging learning over multiple timescales \cite{dagan2024external,peng2024fast}.

In turn, for simpler learning algorithms the answers have been unsatisfactory. \citet{hon1998learning} analyse Nash convergence for the classic fictitious play (FP) algorithm~\citep{brown1951iterative, robinson1951iterative}; however, FP is too sensitive to history and suffers linear regret in adversarial settings. 
\citet{deng2022nash} and \citet{kolumbus2022auctions} study no-regret \emph{mean-based} algorithms~\citep{braverman2018selling}, focusing on both Nash convergence and revenue guarantees. In this case, they provide only \emph{partial} results; \citep{kolumbus2022auctions} show that \emph{if} the distribution of play converges, it implements the same outcome as the second-price auction, whereas \citet{deng2022nash} show that time-average and last-iterate convergence guarantees are contingent on the number of tied buyers (i.e., bidders with the same highest valuation). 

Besides, mean-based learning algorithms are known to be manipulable against a strategic seller~\citep{braverman2018selling}. This has resulted in recent interest in the performance of online gradient ascent (OGA)~\citep{zinkevich2003online}, which is not mean-based \cite{ahunbay2025semicoarse} and not known to be manipulable \cite{kumar2024strategically}. Here, \citet{ahunbay2025semicoarse} leverage recent results on the strong incentive guarantees of OGA \cite{cai2024on,ahunbay2024first} and recover \citep{deng2022nash}'s last-iterate convergence result, yet still require the restrictive assumption that there are at least three highest-value bidders. 

Thus, no prior theoretical result shows that simple no-regret learning dynamics lead to the desirable second-price outcome in first-price auctions; a problem left open in~\citep{deng2022nash, kolumbus2022auctions}. This is in apparent contradiction with empirical work, which suggests that equilibrium is learnable in a wide range of auction settings, from simple single-item auctions to combinatorial auctions with multiple equilibria (e.g. \cite{soda2023,BFHKS21}). This identifies a crucial gap in our knowledge in the analysis of learning dynamics.

\subsection{Technical Barriers \& Contributions}

The culmination of our work in this paper is an unconditional convergence result for online gradient ascent in first-price auctions. Recall that in our model, we suppose buyers are ordered in decreasing integer valuations, $v_1 \geq v_2 \geq \cdots \geq v_N$. 
\begin{theorem*}[Informal Version of \Cref{thm:first-price}]
    Suppose that buyers $1$ and $2$ employ the OGA algorithm in repeated first-price auctions with step sizes ensuring $o(T)$ regret at time $T$, 
    whereas other buyers use arbitrary non-overbidding strategies. 
    Then the long-term outcome of the dynamics satisfies: 
    \begin{enumerate}[label=(\alph*)]
        \item the time-average welfare is at least $v_1 - 2$, and
        \item the time-averaged revenue is at least $v_2 - 2$.
    \end{enumerate}
    Moreover, when $v_1 > v_2 + 2$, the average social welfare converges to $v_1$ and the average revenue converges to $v_2$.
\end{theorem*}
Crucially, our results remove the restrictive assumption on the number of highest-value bidders required in~\citep{deng2022nash, ahunbay2025semicoarse}. Our results demonstrate the robustness of first-price auctions with learning agents: the long-term social welfare and revenue converge to an efficient second-price outcome. Moreover, we in fact require only buyers $1$ and $2$ to employ OGA, with arbitrary (non-overbidding) behaviour from other buyers allowed.

Our results follow from two innovations in the analysis of gradient dynamics in normal-form games, which we believe are applicable in other settings, and thus of independent interest.

\paragraph{Iterative Strategy Elimination \& Online Gradient Ascent}
If we can show that the highest-value buyer (buyer 1) with value $v_1$ bids only $\{v_2-2, v_2-1, v_2\}$ and no lower bids, then we obtain a lower bound of $v_2-2$ on the revenue. The core of our proof is to show that certain bids will be \emph{eliminated}, i.e., the time-average probability of those bids $\rightarrow 0$ as $T \rightarrow \infty$. This procedure is reminiscent of the \emph{iterative elimination of strictly dominated strategies} -- strategies 
with strictly less utility than another strategy no matter what other players do -- which had been applied by \citep{hon1998learning, deng2022nash} in the analysis of mean-based algorithms. Likewise, the conditional result of \citet{kolumbus2022auctions} is based on concluding through the mean-based property that only actions which are \emph{``co-undominated''} are played in a coarse correlated equilibrium.

However, online gradient ascent is \emph{not} guaranteed to eliminate even \emph{strictly} dominated strategies in time-average \citep{ashkenazi2026no}. And since OGA is not mean-based \cite{ahunbay2025semicoarse}, the analysis of OGA requires different techniques. We thus develop a framework in the more general context of normal-form games, which we dub \emph{iterative strategy elimination via potential proofs}. The theoretical foundations of our method is based on recent breakthrough results which show that OGA minimises a stronger regret notion than external regret in concave games, implying convergence to a stronger notion of equilibrium than mere coarse correlated equilibrium~\citep{cai2025new, ahunbay2024first,ahunbay2025semicoarse}. The gist of these results is that OGA minimises regret against strategy modifications induced by gradients of a large class of ``canditate'' potential (more precisely, Lyapunov) functions (c.f. \Cref{sec:refined cce} and \Cref{sec:smooth-statement}). 

By 
finding a suitable potential function $h$ in a normal-form game such that \emph{``vanishing regret against strategy modification induced by $h$''} implies \emph{``some strategies are eliminated''}, we motivate iteratively inspecting smaller subgames. However, prior work does not provide conditions where the iterative argument is valid, i.e., whether the existence of potential functions in each subgame implies that we can conclude all strategies eliminated so far will indeed be played with vanishing probability in time-average as $T \rightarrow \infty$.
Our first technical contribution is to provide sufficient conditions under which this is possible; we show how to compose potential functions across different phases to obtain a single Lyapunov function for the entire game. We present the iterative procedure in \Cref{alg:ievpp-h} and the sufficient condition for its validity in \Cref{thm:iterative}. 

\paragraph{Quadratic Strategy Modifications \& Cubic Potentials} To iteratively eliminate strategies, we need to 
find suitable potential functions; and to do so, we need to know where to look for them! Unfortunately, several natural candidates for potential functions fail to be monotonically changing over non-equilibrium outcomes (c.f. Appendix \ref{sec:potential-non-monotone}). \citet{ahunbay2025semicoarse} characterise the class of linear strategy modifications induced by quadratic potential functions against which OGA incurs vanishing regret. This \emph{semicoarse equilibrium} property is then shown to be useful for first-price auctions with three highest-value bidders.
Unfortunately, it is nevertheless insufficient to conclude convergence to equilibrium when at most two buyers have the highest valuation. However, their results also show that the first-price auction is \emph{critical} when two buyers have highest valuation. In particular, in this case, there exist arbitrarily small perturbations of the bidding space such that OGA is guaranteed to converge to Nash equilibrium.

This criticality suggests that going \emph{just beyond} linear strategy modifications might suffice for convergence analysis of first-price auctions. Therefore, we analyse the class of \emph{cubic} potential functions, which induce a family of quadratic strategy modifications over the probability simplex against which OGA incurs no regret. One major technical challenge here is that, given the \texttt{co-NP}-hardness of deciding whether a matrix is copositive \cite{Dür10,Burer09}, we are unable to obtain a simple characterisation of all cubic potential functions which generate quadratic strategy modifications over the probability simplex. We thus develop a perturbative approach, which allows us to identify a subfamily of such cubic functions. Consequently, we explicitly obtain a stronger regret guarantee and equilibrium refinement for OGA dynamics in normal-form games beyond linear strategy modifications (\Cref{thm:cubic-candidates}).

These cubic potentials are then 
sufficient to 
iteratively eliminate bids in the first-price auction, which yields our time-average convergence result.

To our knowledge, our cubic potential functions induce the first meaningful set of non-linear strategy modifications in normal-form games for equilibrium analysis. That is, our results demonstrate that non-linear notions of regret are relevant 
and may be of greater interest
in the analysis of learning dynamics in normal-form games.

\subsection{Further Related Work}

\paragraph{Learning dynamics in auctions \& games} There is also a vast literature on if learning dynamics converge to equilibrium, and how. The results for gradient-based learning, including not only OGA but also regularised learning algorithms such as Hedge \cite{freund1997decision}, are negative in the general case; including existence of learning cycles \cite{MPP18,PP19}, chaotic behaviour \cite{CP19,AFP21}, and topological impossibility results \cite{MPPS22}. In turn, Nash convergence guarantees are usually proven under assumptions of either monotonicity of utility gradients or existence of a potential function \cite{MZ19,wang2023noregret}. However, various complete-information or Bayesian models of first-price auctions, including ours, satisfy neither condition \citep{bichler2025characterizing,bichler2023convergence} (and implied by \citep{FLN16,ahunbay2024uniqueness}). 

We remark that there is a body of empirical work, which suggests that gradient-based learning algorithms can learn to play equilibrium strategies in a variety of auctions and contests \citep{BFHKS21,soda2023}. However, a theoretical justification has proven elusive even in complete-information and Bayesian auction settings. One line of work analyses the structure of correlated and coarse correlated equilibria to describe learnable outcomes \citep{FLN16,ahunbay2024uniqueness,lopomo2011lp}; the general result is that no-swap regret ensures that outcomes of discretised auctions are close to that of their continuous models, whereas coarse correlated equilibria are only necessarily close to equilibrium in Bayesian first-price auctions when buyers' priors are concave; a restrictive condition. The \emph{semicoarse equilibrium} refinement of \citet{ahunbay2025semicoarse} provides partial convergence results in the complete-information setting, and is our starting point.

\paragraph{Equilibrium computation in first-price auctions} There is also a recent line of work investigating the computational complexity of (approximately) computing equilibria in Bayesian first-price auctions \citep{filosratsikas2021fixp,filos2024computation,filos2025equilibrium,chen2023complexity}. Whereas we would interpret the hardness results as evidence that gradient dynamics should fail to converge quickly (if at all) in the corresponding setting, this is not a concern in our complete-information setting which has an equilibrium we can immediately identify.

\paragraph{Manipulability \& guarantees of gradient ascent} Besides the extensively discussed line of work on local equilibrium theory \citep{cai2024on,cai2025new,ahunbay2024first,ahunbay2025semicoarse}, guarantees of OGA have also attracted attention within the context of manipulability. In general mean-based algorithms are manipulable \citep{braverman2018selling,deng2019strategizing}, and in fact low swap regret is necessary for non-manipulability in arbitrary games. However, \citet{kumar2024strategically,zhao2026no,cai2026online} demonstrate a model of the first-price auction in which any no-external regret algorithm with gradient feedback is non-manipulable. An open question is whether OGA remains non-manipulable in the setting of \citep{braverman2018selling}. We remark their setting with a strategic seller and one buyer is incomparable to ours, as we analyse the setting of competing bidders with a fixed mechanism.

\paragraph{Learning to bid} There has been a long line of work on designing online bidding algorithms for the buyer in repeated auctions, under various settings. See~\citep{balseiro2023contextual, han2020learning, han2025optimal, zhang2022leveraging, wang2023learning, badanidiyuru2023learning, cesa2024role} and references therein for a detailed review. In contrast, our work focuses on the seller's viewpoint and analyses the long-term outcome of online bidding dynamics in auctions.

\paragraph{Autobidding in online ad auctions} Our work is motivated by the practical problem of learning to bid in online advertising auctions, which has been extensively studied as autobidding in the literature. There are various efforts toward designing autobidding algorithms with possibly budget or return on investment (ROI) constraints~\citep{balseiro2019learning, balseiro2020dual, feng2023online, borgs2007dynamics,deng2025no, galgana2025learning}, on analysing autobidding equilibria and their social (liquid) welfare~\citep{conitzer2022multiplicative, conitzer2022pacing, aggarwal2019autobidding, babaioff2020non, deng2024efficiency, liaw2024efficiency}, on the complexity of equilibria or learning in autobidding~\citep{chen2023complexity, chen2024complexity, chen2025constant}. We refer the readers to the recent review by~\citet{aggarwal2024auto} for a comprehensive overview. Most closely related to our work, several studies examine the welfare of learning dynamics in advertising auctions~\citep{gaitonde2022budget, fikioris2023liquid, lucier2024autobidders}, showing that learning dynamics can achieve a constant approximation to the optimal liquid welfare even without converging to equilibria. These results are incomparable to ours as we focus on the standard first-price auction model with online gradient ascent learners.

\subsection{Overview of the Paper}
In \Cref{sec:prelims}, we present our notation and review game theoretic fundamentals, online gradient ascent, and the local equilibrium framework in the context of normal-form games. In \Cref{sec:IEPP}, we introduce our iterative strategy elimination procedure via potential proofs (Algorithm~\ref{alg:ievpp-h}) and establish its validity (Theorem~\ref{thm:iterative}). We then proceed to equilibrium refinement beyond linear strategy modifications in \Cref{sec:cubic}, establishing the no-regret property of OGA against a family of quadratic strategy modifications by identifying a family of cubic potential functions (Theorem~\ref{thm:cubic-candidates}). Finally in \Cref{sec:FPA}, we present our discretised model of the first-price auction as a normal-form game and prove our main result on time-average behaviour (Theorem~\ref{thm:first-price}). We defer details of some proofs (especially those which are repetitive) \& examples to the Appendix.

\section{Preliminaries}\label{sec:prelims}

\newcommand{\ones}{\mathbf{1}}

Throughout the rest of the paper, $\mathbb{N}$ will denote the set of natural numbers (excluding $0$), and for each $N \in \mathbb{N}$, we shall abuse notation and denote by $N$ also the set $\{ m \in \mathbb{N} \ | \ 1 \leq m \leq N\}$.
Given a set $S$, let $\Delta(S)$ be the set of probability distributions on $S$. Then for any $f : S \rightarrow \mathbb{R}$, we shall denote by $\mathbb{E}_{s \sim \sigma}[f(s)]$ the expectation of $f$ when $s$ is drawn from some given $\sigma \in \Delta(S)$. 

We shall follow standard game theoretic convention regarding tuples. Notation for both products of sets and indexed tuples are compactified via dropping the indices, e.g. $X \equiv \times_{i\in I} X_i$ and $s \equiv (s_i)_{i \in I}$. Given an indexed tuple $(s_i)_{i \in I}$ and a proper subset $J \subset I$, we shall write $s_{-J} = (s_i)_{i \in I \setminus J}$, and drop the explicit set notation if $J$ is a singleton; $s_{-j} \equiv (s_i)_{i \in I \setminus \{j\}}$. In turn, if $s, s'$ are both tuples indexed over the set $I$, we shall write $(s'_J, s_{-J})$ for the tuple obtained by replacing the $J$-coordinate entries of $s$ with the corresponding entries of $s'$. This notation extends also to differentiation; if $f : \times_{i \in I} X_i \rightarrow \mathbb{R}$, then $\nabla_i f(x)$ is the gradient of $f$ over the coordinates of $X_i$ only. 

Finally, we shall alternate on the notation for vectors and matrices depending on convenience; for instance, $v^T w$ and $\bilin{v}{w}$ will both be used for the inner product of two vectors, and we will expand the sums whenever explicit calculations are required. For a real vector $v \in \mathbb{R}^D$, we will write $v \geq \alpha$ ($v \leq \alpha$) if the entries of $v$ are all upper (lower) bounded by $\alpha \in \mathbb{R}$. We will use $\ones = (1, \ldots, 1)$ to denote the all-ones vector of appropriate size.
We will use $\mathbb{I}[\cdot]$ for the proposition valued function which equals $1$ if the proposition is true and $0$ otherwise.

\subsection{Normal-form games, mixed-extensions, and notions of equilibrium}

We first begin with a review of game theoretic fundamentals.

\begin{definition}
    A \textbf{(finite) normal-form game} $\Gamma$ is a tuple $(N,(A_i)_{i \in N}, (u_i)_{i \in N})$, where: (1) $N$ is the number (and by choice of notation, the set) of players, (2) for each player $i$, $A_i$ is a finite set of actions of player $i$, where we denote $A \equiv \times_{i \in N} A_i$ as the set of action profiles or outcomes of the game, and (3) each $u_i : A \rightarrow \mathbb{R}$ is the utility (or payoff) function of player $i$. We shall denote by $D_i = |A_i|$ the number of actions available to each player, and let $D = \sum_{i \in N} D_i$.
\end{definition}

A normal-form game provides a simple model of strategic interaction between agents. Each player $i$ chooses an action $a_i \in A_i$. 
Given the outcome $a = (a_1,a_2,...,a_N) \in A$, each player $i$ obtains a payoff $u_i(a)$. The model accounts for \emph{``pure''} strategies, but can be readily extended to allow for randomisations in actions by players; if each player $i$ independently chooses their actions following a distribution $x_i \in \Delta(A_i)$, then we may compute their utilities in expectation. Extending the action sets and payoff functions in this manner provides the \emph{mixed-extension} of a normal-form game.

\begin{definition}
    Let $\Gamma = (N,(A_i)_{i \in N}, (u_i)_{i \in N})$ be a finite-normal form game. Then its \textbf{mixed-extension} is a tuple $(N,(X_i)_{i \in N}, (\hat{u}_i)_{i \in N})$, where 
    $X_i = \Delta(A_i) = \{x_i \in \mathbb{R}^{A_i} \ | \  x_i \geq 0, 
    \ones^T x_i = 1
    \}$ is the set of probability distributions over $A_i$. Let $X = \times_{i \in N} X_i$ be the set of mixed-strategy profiles. 
    The utility function $\hat{u}_i : X \rightarrow \mathbb{R}$ is obtained by extending $u_i$ via expectation; for any $x \in X$,
    $$ \hat{u}_i(x) ~ = ~ \mathbb{E}_{\forall j \in N, a_j \sim x_j}[u_i(a)] ~ = ~ \sum_{a \in A} \left(\prod_{j \in N} x_j(a_j)\right)\cdot u_i(a).$$
\end{definition}

\begin{notation*}
    For simplicity, we shall abuse notation and write $u_i$ to also mean  $\hat{u}_i$. 
\end{notation*}

Game theoretic analysis often consists of modelling some setting (of practical or theoretical ``relevance'') as a game, and then asserting that certain outcomes should be considered to be ``expected'' owing to their stability if agents are taken to be utility maximising. A notion of equilibrium is a formalisation of any such concept of stability, and in the context of (mixed-extensions of) normal-form games, perhaps the most commonly considered one is that of a Nash equilibrium. 

\begin{definition}
A mixed-strategy profile $x^* \in X$ is called a \textbf{Nash equilibrium} if for every player $i \in N$ and any mixed-strategy $x_i \in X_i$ of player $i$, $u_i(x_i,x^*_{-i}) \leq u_i(x^*)$. 
\end{definition}

The mixed-extension of any normal-form game has a Nash equilibrium \cite{Nash50}. An outcome of the mixed-extension, of course, induces a probability distribution over the set of outcomes of the underlying normal-form game, in which each player's action is drawn independently following the player's mixed-strategy. Allowing for correlations in this distribution provides the following two generalisations of a Nash equilibrium.

\begin{definition}
    Let $\sigma \in \Delta(A)$, then $\sigma$ is called a \textbf{correlated equilibrium} (CE) if for every player $i \in N$ and every pair of 
    actions $a_i, a'_i \in A_i$, $\mathbb{E}_{a_{-i} \sim \sigma(a_i,\cdot)}[u_i(a'_i,a_{-i})-u_i(a)] \leq 0$. If instead we have $\mathbb{E}_{a \in \sigma}[u_i(a'_i,a_{-i})-u_i(a)] \leq 0$ for every player $i \in N$ and every $a'_i \in A_i$, then $\sigma$ is called a \textbf{coarse correlated equilibrium} (CCE).
\end{definition}

\subsection{Gradient ascent, local coarse correlated equilibrium, and time-average guarantees}

We shall consider the case when (some) players repeatedly participate in a (mixed-extension of a) normal-form game, and update their strategies using \textbf{online gradient ascent} \cite{zinkevich2003online}. That is, a player $i \in N$ starts with some mixed strategies $x_i^1$, and at each time period $t$ updates their strategies for the next time period by setting
$$ x_i^{t+1} = \proj{X_i}{x_i^t + \eta_{it} \nabla_i u_i(x^t)},$$
where $\proj{C}{x} \equiv \arg \min_{x' \in C} \|x-x'\|$ denotes the projection of $x$ onto a (non-empty) closed \& convex set $C$, and $\eta_{it}$ is assumed (\emph{throughout the paper}) to be some non-increasing sequence of step sizes. We use $\|\cdot\|$ to denote the Euclidean ($\ell_2$) norm. 

Over time periods $t = 1,2,...,T$, we will denote by $\sigma(a) = (1/T) \cdot \sum_{t = 1}^{T} \prod_{j \in N} x^t_j(a_j)$ the \textbf{time-average distribution of play}. It follows from \cite{zinkevich2003online} that if every player runs online gradient ascent with suitable step sizes, then their \emph{external regret} is bounded $o(T)$. For instance, if $\eta_{it} \propto 1/\sqrt{T}$ or $\propto 1/\sqrt{t}$, then for any action $a'_i \in A_i$,
$$ \sum_{a \in A} \sigma(a) \cdot \left[ u_i(a'_i,a_{-i}) - u_i(a) \right] \leq O(1/\sqrt{T}).$$
As a consequence, 
the time-average distribution of play is an approximate coarse correlated equilibrium. However, it is known that online gradient ascent does not guarantee convergence to a correlated equilibrium in this manner in general \cite{viossat2015evolutionary,ahunbay2025semicoarse}. 

Recent work \cite{cai2024on,ahunbay2024first,cai2025new} shows that online gradient ascent actually minimizes a stronger notion of regret than external regret (see \Cref{sec:refined cce} for a more detailed discussion). As a result, online gradient ascent dynamics converge to a refined subset of coarse correlated equilibria in normal-form games.
Moreover, these results apply to the wider class of (non-concave) \emph{smooth games} (c.f. Appendix \ref{sec:smooth-statement}). We may state these refined guarantees in terms of either $\Phi$-regret~\citep{greenwald2003general} guarantees against strategy modifications~\citep{cai2024on, cai2025new}) or their first-order limit~\citep{ahunbay2024first}. In this paper, we work within the latter formalism, which we summarise here within the scope of mixed-extensions of normal-form games. 

For what follows, given a convex set $C$, we denote by $\TC{C}{x}$ the \textbf{tangent cone} to $C$ at $x$, which is the closure of the pointed cone $\{ \gamma (y-x) \ | \ y \in C, \gamma \geq 0\}$. Meanwhile, $\NC{X}{x}$ will denote the \textbf{normal cone}, the dual cone to $\TC{X}{x}$, defined as $\NC{X}{x} = \{ y \ | \ y^T z \leq 0, \ \forall \ z \in \TC{X}{x}\}$. 

\begin{definition}\label{def:lcce}
    A distribution $\sigma \in \Delta(X)$ is called an $\epsilon$\textbf{-local correlated equilibrium} with respect to a set $F$ of magnitude bounded, Lipschitz continuous vector fields $f : X \rightarrow \mathbb{R}^D$ if for every $f \in F$,
    $$ \sum_{i \in N} \mathbb{E}_{x \sim \sigma}\left[\bilin{\proj{\TC{X_i}{x_i}}{f_i(x)}}{\nabla_i u_i(x)} \right] ~ \leq ~ \epsilon \cdot \poly(\vec{u},G_f,L_f).$$
    Here, $\vec{u}$ is the tuple of players' payoffs for outcomes $(u_i(a))_{i \in N, a \in A}$, whereas $G_f$ and $L_f$ are respectively the bounds on the magnitude and the Lipschitz modulus of $f$, i.e. $\|f(x)\| \leq G_f$ for any $x \in X$ and $\|f(x)-f(y)\| \leq L_f \|x-y\|$ for any $x,y \in X$.
    The equilibrium is called \textbf{coarse} if $F$ consists solely of gradient fields, i.e. for every $f \in F$, $f = \nabla h$ for some function $h : X \rightarrow \mathbb{R}$.
\end{definition}

In the analysis of gradient ascent, we often consider the case where (1) the projection onto the tangent cone $\proj{\TC{X_i}{x_i}}{f_i(x)}$ is not required, and (2) $f : X \rightarrow \mathbb{R}^D$ is a gradient field $f = \nabla h$ with bounded and Lipschitz continuous $\nabla h$. 
Such $h$ will be ``suitable'' potential functions to use. 

\begin{definition}\label{def:tangency}
    Let $C \subseteq \mathbb{R}^D$ be a convex set. A vector field $f : C \rightarrow \mathbb{R}^D$ is called \textbf{tangent} if for every $x \in C$, $f(x) \in \TC{C}{x}$. 
    A function $h : C \rightarrow \mathbb{R}$ is called \textbf{tangential} if $\nabla h$ is tangent. We call a function $h : C \rightarrow \mathbb{R}$ \textbf{Lipschitz-differentiable} if $\nabla h$ admits a Lipschitz modulus $L_h$.\footnote{For normal-form games, $X$ is compact, and thus a bound $G_h$ on $\|\nabla h\|$ is implied by Lipschitz continuity of $\nabla h$.} 
\end{definition}

Intuitively, when $C = \times_{i \in N} \Delta(A_i) = X$ and if $f : X \rightarrow \mathbb{R}^D$ is tangent then for each mixed-strategy profile $x$ and for each player $i$,  the vector field $f_i$ must ``point inwards'' towards the probability simplex $\Delta(A_i)$. In fact, \citet{ahunbay2025semicoarse} show that if a Lipschitz continuous vector field $f : X \rightarrow \mathbb{R}^D$ is tangent, then there exists $\delta > 0$ such that $x + \delta f(x) \in X$ for every mixed-strategy profile $x$. This implies that $x_i + \delta f_i(x)$ must itself be a probability distribution over $A_i$ for every player $i$. This is assured if $f$ satisfies, for any player $i$ and any $x \in X$,
\begin{align}
    \ones^T f_i(x) & = 0 \quad \textnormal{and} \label{cons:probability} \\
    f_i(x)_{a_i} & \geq 0 \quad \forall \ i \in N, a_i \in A_i \text{ such that } x_i(a_i) = 0. \label{cons:tangency}
\end{align}
The first constraint (\ref{cons:probability})
ensures that the probabilities sum up to 1: $\sum_{a_i \in A_i} \big( x_i(a_i) + \delta f_i(x)_{a_i} \big) = 1$.
The second set of constraints (\ref{cons:tangency}) assures that no action is assigned a negative probability, 
so $x_i(a_i) + \delta f_i(x)_{a_i} \geq 0$.

To then provide intuition for local equilibrium, we remark that by the multilinearity of utilities, 
\begin{equation}\label{eq:util-grad}
u_i(x) = \sum_{a \in A} \left(\prod_{j \in N} x_j(a_j)\right)u_i(a) \quad \implies \quad \frac{\partial u_i(x)}{\partial x_i(a_i)} = \sum_{a_{-i} \in A_{-i}} \left(\prod_{j \neq i} x_j(a_j)\right) u_i(a).
\end{equation}Therefore, if $f$ is tangent, then for small enough $\delta > 0$,
\begin{align*}
& \bilin{f_i(x)}{\nabla_i u_i(x)}  = \sum_{a_i \in A_i} \left[ \sum_{a_{-i} \in A_{-i}} \left( \prod_{j \neq i} x_j(a_j) \right) u_i(a) \right] \cdot f_i(x)_{a_i} \\ 
& = \sum_{a \in A} \left( \prod_{j \neq i} x_j(a_j) \right) \frac{u_i(a) \cdot \big( x_i (a_i)+\delta f_i(x)_{a_i} - x_i(a_i) \big)}{\delta} = \frac{u_i(x_i + \delta f_i(x), x_{-i}) - u_i(x)}{\delta}. 
\end{align*}
Definition \ref{def:lcce} thus has the following interpretation when the tangent vector field $f : X \rightarrow \mathbb{R}^D$ is zero valued except for some player $i$; over the distribution of play, player $i$ should have low regret ($\leq \delta \epsilon \cdot \poly(\vec{u},G_f,L_f)$) against strategy deviations $x_i \mapsto x_i + \delta f_i(x)$.

Both CCE and CE of a normal-form game are local CE of its mixed-extension \citep{ahunbay2024first} for appropriately defined set of tangent vector fields over $X$; the former is with respect to a set of gradient fields, whereas the latter is represented by a set of non-conservative vector fields. The fact that online gradient ascent has vanishing external regret guarantees but fails to do so for internal/swap regret then turns out to be no accident, as \citep{ahunbay2024first,ahunbay2025semicoarse,cai2025new} together delineate the adversarial guarantees of gradient ascent to be precisely against gradient fields of tangential functions; we cite below the simple bound of \cite{ahunbay2024first}. 

\begin{proposition}[\cite{ahunbay2024first}, Theorem B.1]\label{prop:reg-guarantee-gd}
    Suppose that a subset $\hat{N}$ of players implement online gradient ascent with the same, non-increasing step sizes $\eta_t$. Then for any Lipschitz-differentiable tangential function $h : \times_{i \in \hat{N}} X_i \rightarrow \mathbb{R}$, ``aggregate regret'' against $\nabla h$ among the players in $\hat N$ satisfies
    $$ \sum_{t = 1}^T\sum_{i \in \hat{N}} \bilin{\nabla_i h(x^t_{\hat{N}})}{\nabla_i u_i(x^t)} ~ \leq ~ \sum_{t=1}^T \frac{h(x_{\hat{N}}^{t+1}) - h(x_{\hat{N}}^{t})}{\eta_t} + 2 \eta_t L_h \bigg( \sum_{i \in \hat{N}} G_i \bigg)^2,$$
    where $G_i \geq 0$ is any bound\footnote{We may take this to equal $D_i \cdot  \max_{a \in A, a'_i \in A_i} \{u_i(a) - u_i(a'_i,a_{-i})\}$, c.f. Lemma \ref{lem:GL-bounds}.} on $\|\nabla_i u_i\|$, and $L_h$ is the Lipschitz modulus of $\nabla h$. As a consequence,  
    \begin{equation}\label{eq:asymptotic-regret}
        \sum_{t = 1}^T\sum_{i \in \hat{N}} \bilin{\nabla_i h(x^t_{\hat{N}})}{\nabla_i u_i(x^t)} ~ \leq ~ \Bigg( 2\hat{N} G_h + 2 L_h \bigg( \sum_{i \in \hat{N}} G_i \bigg)^2\Bigg) \Bigg( \frac{1}{\eta_T} + \sum_{t=1}^T \eta_t \Bigg).
    \end{equation}
\end{proposition}
Proposition \ref{prop:reg-guarantee-gd} implies, through (\ref{eq:asymptotic-regret}), that if step sizes are chosen $\eta_t \sim 1/\sqrt{T}$ or $\sim 1/\sqrt{t}$, then the aggregate regret $\sum_{t = 1}^T\sum_{i \in \hat{N}} \langle{\nabla_i h(x^t_{\hat{N}})},{\nabla_i u_i(x^t)}\rangle = O(\sqrt{T})$. Local equilibrium guarantees then follow; for instance, if all players use the same, suitably declining step sizes, the distribution which samples from $(x^t)_{t=1}^T$ uniformly at random forms an approximate local coarse correlated equilibrium with respect to the set of gradient fields of all tangential\footnote{Guarantees against gradient fields of non-tangential functions are also known, but outside the scope of our discussion.} functions $h : X \rightarrow \mathbb{R}$. The derivation of (\ref{eq:asymptotic-regret}) follows from standard arguments in learning theory bounding approximately telescoping sums (e.g. Theorem 2.3, \cite{cesa2006prediction}), which we elaborate on in Appendix \ref{sec:extend-reg-gd}.

\citet{ahunbay2024first} shows that the term $\sum_{i \in \hat{N}} \langle{\nabla_i h(x^t_{\hat{N}})},{\nabla_i u_i(x^t)}\rangle$ in Proposition \ref{prop:reg-guarantee-gd} may be interpreted as the rate of change of the value of $h$ for the unprojected continuous gradient dynamics of the game; if $dx_i(t)/dt = \nabla_i u_i(x(t))$, then
$$ \frac{dh(x(t))}{dt} ~ = ~ \sum_{i \in N} \bilin{\nabla_i h(x_{\hat{N}}(t))}{\frac{dx_i(t)}{dt}} ~ = ~ \sum_{i \in \hat{N}} \bilin{\nabla_i h(x_{\hat{N}}(t))}{\nabla_i u_i(x(t))}.$$
Moreover, it constitutes a lower bound for the rate of change for the projected dynamics (c.f. Appendix \ref{sec:cont-dyn}). For this reason, we shall refer to $\sum_{i \in \hat{N}} \bilin{\nabla_i h(\cdot)}{\nabla_i u_i(\cdot)} :  X \rightarrow \mathbb{R}$ as the \textbf{(position dependent) rate of change} of $h$. 

To conclude this section, we remark that Proposition \ref{prop:reg-guarantee-gd} provides a ``recipe''
to prove the time-average convergence of some sequence $(q(x^t))_{1 \leq t \leq T}$ when (some) players employ online gradient ascent:
if we can find a (suitably additively separable) tangential function $h$ whose rate of change satisfies $\sum_{i \in N} \bilin{\nabla_i h(x)}{\nabla_i u_i(x)} \geq q(x) - \gamma$ for any $x \in X$ and some $\gamma$,
then we can obtain an upper bound on $\frac{1}{T} \sum_{t=1}^T q(x^t)$. This is formalised below:  

\begin{proposition}\label{prop:performance-guarantees}
    Let $N = N_* \cup (\cup_{\alpha \in I} N_\alpha)$ be a partition of $N$ for some index set $I$. Consider the setting where, over time periods $t = 1,2,...,T$, (1)  players $i \in N_*$ choose their mixed-strategies arbitrarily, and (2) each player $i \in N_\alpha$ employs online gradient ascent with the same non-increasing step sizes $(\eta_{t\alpha})_{t =1}^T$.
    Let $q : X \rightarrow \mathbb{R}$ be any function. If there exist Lipschitz-differentiable tangential functions $(h_\alpha : \times_{i \in N_\alpha} X_i \rightarrow \mathbb{R})_{\alpha \in I}$ and $\gamma \in \mathbb{R}$ such that, setting $h(x) = \sum_{\alpha \in I} h(x_{N_\alpha})$, the rate of change of $h$ satisfies \begin{equation}\label{eq:lyapunov-condition}
        \gamma + \sum_{i \in N} \bilin{\nabla_i h(x)}{\nabla_i u_i(x)} ~ \geq ~ q(x) \quad \forall \ x \in X,
    \end{equation}
    then the time-average of $q(x^t)$ is bounded,\begin{equation}\label{eq:time-avg-bound}\frac{1}{T} \sum_{t = 1}^T q(x^t) ~ \leq ~ \gamma +  \sum_{\alpha \in I} \frac{1}{T} \left(  \frac{1}{\eta_{T\alpha}} + \sum_{t=1}^T \eta_{t\alpha} \right)\Bigg( 2 N_\alpha G_h + 2 L_h \bigg( \sum_{i \in N_\alpha} G_i \bigg)^2\Bigg).\end{equation}
\end{proposition}

Proposition \ref{prop:performance-guarantees} is a variant of time-average guarantees provided in \cite{ahunbay2024first}, obtained by time-averaging both sides of \eqref{eq:lyapunov-condition}.
The separability follows since $h$ is guaranteed to be tangential due to the product structure on $X$; since $X = \times_{i \in N} X_i$, $\TC{X}{x} = \times_{i \in N} \TC{X_i}{x_i}$.

\subsection{Refined analysis in normal-form games \& semicoarse equilibrium}\label{sec:refined cce}
Proposition \ref{prop:reg-guarantee-gd} also provides a recipe to refine the analysis of online gradient ascent in normal-form games, where we restrict attention to a class of tangential functions $h : X \rightarrow \mathbb{R}$ and analyse the ``partial fragment'' of the resulting local equilibrium guarantees. A recent result of \cite{ahunbay2025semicoarse} in normal-form games 
shows that, even within the class of strategy modifications which are linear endomorphisms $\Delta(A_i) \rightarrow \Delta(A_i)$, online gradient ascent provides strictly stronger guarantees than no-external regret. Specifically, by considering the tangential function $h(x) = \frac{1}{2} x_i^T Q_i x_i + q_i^T x_i$ for some player $i$, and a suitable pair $(Q_i,q_i)$ of a symmetric matrix $Q_i$ and a vector $q_i$, one can prove 
that online gradient ascent has a no-regret guarantee against any affine-symmetric linear strategy modification \citep{ahunbay2025semicoarse,cai2025new}. In the context of (mixed extensions of) normal-form games, in the limit $T \rightarrow \infty$ with suitably decreasing step sizes, this stronger no-regret property can be expressed via finitely many linear inequalities over the time-average distribution of play:

\begin{proposition}[\cite{ahunbay2025semicoarse}, with bounds of Proposition \ref{prop:reg-guarantee-gd}]\label{prop:scce-constraints}
    Suppose in the mixed-extension of a normal-form game that player $i$ updates strategies via online gradient ascent with non-increasing step sizes. Then the time-average distribution of play $\sigma$ satisfies:
    \begin{enumerate}
        \item For any proper subset $S_i$ of actions $A_i$, player $i$ incurs low regret against deviating from actions in $S_i$ to the uniform distribution on $A_i \setminus S_i$, explicitly,
        $$ \sum_{a_{-i} \in A_{-i}} \sum_{a \in S_i} \sigma(a) \cdot \Bigg[ \sum_{a'_i \in A_i \setminus S_i} \frac{u_i(a'_i,a_{-i})}{|A_i\setminus S_i|} - u_i(a) \Bigg] ~ \leq ~ \frac{8}{T\eta_{iT}} + \frac{\sum_{t=1}^T \eta_{it}}{T} \cdot 4G_i^2.$$
        \item For any cycle $C$ of actions $(a_{i1}, a_{i2}, ..., a_{ik})$, where $a_{i(k+1)} = a_{i1}$, player $i$ incurs low regret against deviating from $a_{i\ell}$ to the uniform distribution on $a_{i(\ell-1)},a_{i(\ell+1)'}$, explicitly, 
        $$ \sum_{a_{-i} \in A_{-i}} \sum_{\ell = 1}^k \sigma(a) \cdot \left[ \frac{u_i(a_{i(\ell-1)},a_{-i})+u_i(a_{i(\ell+1)},a_{-i})}{2} - u_i(a) \right] ~ \leq ~ \frac{8}{T\eta_{iT}} + \frac{\sum_{t=1}^T \eta_{it}}{T} \cdot 6G_i^2.$$
    \end{enumerate}
    Moreover, any linear strategy modification against which projected gradient ascent incurs no-regret is necessarily a conic combination of the above two classes of strategy modifications.
\end{proposition}

Important for our analysis, \citet{ahunbay2025semicoarse} show that for each $S_i \subset A_i$,  Proposition \ref{prop:scce-constraints}.(1) follows from the local equilibrium constraint for the function 
\begin{equation}\label{eq:good-devs}
    h_i^{S_i}(x_i) ~ = ~ -\frac{1}{2|A_i\setminus S_i|} \Bigg( 1- \sum_{a_i \in A_i\setminus S_i} x_i(a_i) \Bigg)^2 - \frac{1}{2} \sum_{a_i \in S_i} x_i(a_i)^2.
\end{equation}
It will be helpful for our later analysis to illustrate how this is the case. First, note that for any $a_i \in S_i$, $\partial h_i^{S_i}(x_i)/\partial x_i(a_i) = -x_i(a_i)$. For $a_i \in A_i \setminus S_i$,
$$ \frac{\partial h_i^{S_i}(x_i)}{\partial x_i(a_i)} ~ = ~ \frac{-1}{2|A_i \setminus S_i|} \cdot 2\Bigg(1-\sum_{a'_i \in A_i \setminus S_i} x_i(a'_i) \Bigg) \cdot (-1) ~ = ~ \frac{1}{|A_i \setminus S_i|} \sum_{a'_i \in S_i} x_i(a'_i),$$
where we use the fact that $\sum_{a'_i \in A_i} x_i(a'_i) = 1$. Thus,
\begin{align}
    \sum_{a_i\in S_i} \frac{\partial h_i^{S_i}(x_i)}{\partial x_i(a_i)} \frac{\partial u_i(x)}{\partial x_i(a_i)} ~ & = ~ \sum_{a_i \in S_i} -x_i(a_i) \sum_{a_{-i} \in A_{-i}} \Bigg( \prod_{j \neq i} x_j(a_j)\Bigg) u_i(a), \quad \text{and} \label{eq:7} \\
    \sum_{a_i\in A_i \setminus S_i} \frac{\partial h_i^{S_i}(x_i)}{\partial x_i(a_i)} \frac{\partial u_i(x)}{\partial x_i(a_i)} ~ & = ~ \sum_{a_i \in A_i \setminus S_i} \Bigg( \frac{1}{|A_i \setminus S_i|} \sum_{a'_i\in S_i} x_i(a'_i) \Bigg) \sum_{a_{-i} \in A_{-i}}\Bigg( \prod_{j \neq i} x_j(a_j)\Bigg) u_i(a) \nonumber \\
    & = ~ \sum_{a_i \in S_i} \sum_{a_{-i} \in A_{-i}} \Bigg( \prod_{j \in N} x_j(a_j)\Bigg) \Bigg( \frac{1}{|A_i \setminus S_i|} \sum_{a'_i \in A_i \setminus S_i} u_i(a'_i,a_{-i})\Bigg). \label{eq:8}
\end{align}
Here, the last line follows from exchanging the use of $a'_i$ and $a_i$ as dummy variables, and rearranging the sum. Combining (\ref{eq:7})$+$(\ref{eq:8}),  the time average of the rate of change of $h_i^{S_i}$ equals
\begin{equation}\label{eq:deviation-form}
\bilin{\nabla_ih_i^{S_i}(x_i)}{\nabla_iu_i(x)} = \sum_{a_{-i} \in A_{-i}} \sum_{a_i \in S_i} \Bigg( \prod_{j \in N} x_j(a_j) \Bigg) \Bigg( \frac{1}{|A_i \setminus S_i|} \sum_{a'_i \in A_i \setminus S_i} u_i(a'_i, a_{-i}) - u_i(a) \Bigg).
\end{equation}
Recalling that $\sigma(a) = \frac{1}{T}\sum_{t = 1}^T \prod_{j \in N} x^t_j(a_j)$, 
we obtain the 
\emph{form of the deviations} on the left-hand-side of Proposition \ref{prop:scce-constraints}.(1). The regret bounds then follow from Proposition \ref{prop:reg-guarantee-gd}, as we observe that $h_i^{S_i}$ has a bound $2$ on the magnitude of its gradient and $1$ on its Lipschitz modulus.

\section{Iterative Elimination of Strategies for Online Gradient Ascent
via Potential Proofs
}\label{sec:IEPP}

Our central argument in
the analysis of online gradient ascent in first-price auctions
will be to \emph{iteratively eliminate strategies via potential proofs}.
That is, we shall establish that if some players implement projected gradient ascent (with suitable, possibly unequal step sizes), as $T \rightarrow \infty$ some strategies will be played at a vanishing rate. This will motivate us to look at the reduced game where these strategies are removed from the players' 
sets of actions. Iteratively, we thus shall analyse subgames with action sets $(A_i^k)_{i \in N}$ for $1 \leq k \leq K$, where $A_i^k \supseteq A_i^\ell$ for $k < \ell$, with strict inclusion for at least one player $i$. 

At each step of the iteration, we shall establish that a given set of outcomes is played with low probability as follows. Suppose that in the initial step, we want to argue that each player $i$ plays actions $a_i \in A'_i \subseteq A_i$ 
with low probability in time-average. By the definition of time-average distribution of play $\sigma(a)$, the probability that some player $i$ plays an action $a_i \in A_i'$ equals 
\begin{align*}
    \sum_{a \in A} \sigma(a) \cdot \mathbb{I}[\exists i \in N,a_i \in A'_i] ~ & = ~ \sum_{a \in A} \left( \frac{1}{T} \sum_{t = 1}^T \prod_{j \in N} x^t_j(a_j) \right)\mathbb{I}[\exists i \in N,a_i \in A'_i] \\ & = ~ \frac{1}{T} \sum_{t = 1}^T \left[ \sum_{a \in A}\bigg(\prod_{j \in N} x^t_j(a_j) \bigg)\mathbb{I}[\exists i \in N,a_i \in A'_i] \right] ~ \equiv ~ \frac{1}{T} \sum_{t = 1}^T q(x^t).
\end{align*}
In other words, the time-average probability that some player $i$ plays an action $a_i \in A_i'$ is the time-average value of the function $q : X \rightarrow \mathbb{R}$ defined by
\[ q(x) = \bigg(\prod_{j \in N} x_j(a_j) \bigg)\mathbb{I}[\exists i \in N,a_i \in A'_i].\]
Proposition \ref{prop:performance-guarantees} then tells us how to upper bound this function; if we can find a tangential function $h : X  \rightarrow \mathbb{R}$ such that 
\begin{enumerate}[label=(\alph*)]
    \item for a partition $N = N_* \cup (\cup_{\alpha \in I}N_\alpha)$, $h(x) = \sum_{\alpha \in I} h_\alpha(x_{N_\alpha})$ is additively separable into Lipschitz-differentiable tangential functions $h_\alpha : \times_{i \in N_\alpha} X_i \rightarrow \mathbb{R}$, and
    \item the \emph{``Lyapunov condition''}\footnote{This condition implies that, for the continuous time gradient dynamics of the game, the value of $h$ is weakly increasing at each $x$, and strictly so whenever an action $a_i \in A'_i$ is played with positive probability -- c.f. Appendix \ref{sec:cont-dyn} and (\ref{eq:cont-time-argument}).} (\ref{eq:lyapunov-condition}) holds for $q$ and $h$ with $\gamma = 0$, that is, 
    \begin{equation}\label{cond:lyapunov}
    \sum_{i \in N}  \bilin{\nabla_i h(x)}{\nabla_i u_i(x)} ~ \geq ~  \sum_{a \in A} \bigg( \prod_{j \in N} x_j(a_j) \bigg)\mathbb{I}[\exists i \in N, a_i \in A'_i] \qquad \forall \ x \in X,
    \end{equation}
\end{enumerate} 
then we immediately obtain a time-average guarantee whenever for each $\alpha \in I$, each player in $N_\alpha$ employs online gradient ascent with the same step sizes. 

\begin{proposition}\label{prop:time-avg-guarantees}
    Let $N = N_* \cup (\cup_{\alpha \in I} N_\alpha)$ be a partition of $N$ for some index set $I$. Consider the setting where, over time periods $t = 1,2,...,T$, players $i \in N_*$ choose their mixed strategies arbitrarily, and each player $i \in N_\alpha$ employs online gradient ascent with the same non-increasing step sizes $(\eta_{t\alpha})_{t =1}^T$. Then, if there exist Lipschitz-differentiable tangential functions $(h_\alpha : \times_{i \in N_\alpha} X_i \rightarrow \mathbb{R})_{\alpha \in I}$ such that, setting $h(x) = \sum_{\alpha \in I} h(x_{N_\alpha})$, $h$ satisfies (\ref{cond:lyapunov}) with $\gamma = 0$, then for the time-average distribution $\sigma$,
    \begin{equation}\label{eq:prob-time-avg-bound}\sum_{a \in A} \sigma(a) \mathbb{I}[\exists i \in N, a_i \in A'_i] ~ \leq ~ \sum_{\alpha \in I} \frac{1}{T} \Bigg( \frac{1}{\eta_{T\alpha}} + \sum_{t=1}^T \eta_{t\alpha} \Bigg)\Bigg( 2 N_\alpha G_h + 2 L_h \bigg( \sum_{i \in N_\alpha} G_i \bigg)^2\Bigg).\end{equation}
\end{proposition}

A particularly general guarantee is obtained when $h$ 
is additively separable
and dependent only on the mixed strategies of a subset $\tilde{N}$ of players; if $h(x) = \sum_{i \in \tilde{N}} h_i(x_i)$, then we obtain a time-average guarantee that the actions $a_i \in A'_i$ are played with low probability when each player $j \in \tilde{N}$ employs projected gradient ascent with suitably declining, \emph{potentially distinct} step sizes, \emph{whereas players $i \notin \tilde{N}$ may play arbitrarily.}

We thus proceed to formally define this procedure for deducing, through some \emph{``potential''} (or \emph{Lyapunov function}) $h$, 
that a player $i$ should play some of their actions $A'_i$ with vanishing probability; as well as the iterative application of said procedure. At the $\ell$-th step, this process will result in a \textbf{subgame} $\Gamma^\ell = (N,(A_i^\ell)_{i \in N},(u_i)_{i \in N})$ where each player has access to a subset of actions $A_i^\ell \subseteq A_i$. 

\begin{definition}
    \label{eq:epp}
    Let $\Gamma = (N, (A_i)_{i \in N},(u_i)_{i \in N})$ be a normal-form game. Let $A'_i \subsetneqq A_i$ for each player $i \in N$,
    with at least one $A'_i$ non-empty. Then $(A'_i)_{i \in N}$ are \textbf{eliminated via a potential proof by $h$} if $h : \times_{i \in N} \Delta(A_i) \rightarrow \mathbb{R}$ is Lipschitz-differentiable, tangential, and satisfies  (\ref{cond:lyapunov}) for $(A'_i)_{i \in N}$.
\end{definition}

\begin{algorithm}[H]
    \caption{Iterative elimination via potential proofs (IEPP)}\label{alg:ievpp-h}
    \SetKwComment{Comment}{/* }{ */}
    \SetKwInput{Input}{Input}
    \SetKwInput{Output}{Output}
    \Input{A normal-form game $\Gamma = (N,(A_i)_{i \in N}, (u_i)_{i \in N}) = \Gamma^0$.}
    \Output{Subgame $\Gamma^K=(N,(A_i^K)_{i \in N},(u_i)_{i \in N})$, finite sequence of functions $(h^\ell)_{0 \leq \ell < K}$.}
    Initialise $\ell \leftarrow 0$\;
    Initialise $A_i^0 = A_i$ for all $i \in N$\;
    \While{``able to repeat''}{
        For game $\Gamma^\ell$, find $h^{\ell}$ which eliminates not-all-empty $(A^{'\ell}_i)_{i \in N}$ via a potential proof\;
        Set $A_i^{\ell+1} = A_i^\ell \setminus A^{'\ell}_i$ for every $i \in N$\;
        Set $\Gamma^{\ell+1}=(N,(A^{\ell+1}_i)_{i \in N}, (u_i)_{i \in N})$, with each $u_i$ restricted to $\times_{i \in N}A^{\ell+1}_i$\;
        Set $\ell \leftarrow \ell +1$\;
    }
    Set $K \leftarrow \ell$\;
\end{algorithm}

Running Algorithm \ref{alg:ievpp-h}, we obtain a finite sequence of functions $(h^\ell)_{0 \leq \ell < K}$. For our iterative elimination process to be valid, we want to be able to argue that on time-average, each player $i$ will play actions in $\cup_{\ell = 0}^{K-1} A^{'\ell}_i$ with low probability. By Proposition \ref{prop:time-avg-guarantees}, a sufficient condition is for the functions $(h^\ell)_{0 \leq \ell < K}$ to \emph{iteratively assemble into a Lyapunov function}.
In principle, each $h^\ell$ is defined on some $\times_{i \in \hat{N}} \Delta(A_i^\ell)$ only, and must be first extended to some $\hat{h}^\ell : \times_{i \in \hat{N}} \Delta(A_i) \rightarrow \mathbb{R}$ on the entire action set.
Then, if there exist multipliers $\epsilon(0), \epsilon(1), ..., \epsilon(K-1) \in \mathbb{R_{+}}$ such that (\ref{cond:lyapunov}) holds for $h = \sum_{\ell = 0}^{K-1} \epsilon(\ell) \hat{h}^\ell$ and $(\cup_{\ell = 0}^{K-1} A^{'\ell}_i)_{i \in N}$, then we can use Proposition \ref{prop:time-avg-guarantees} with $h$ to prove that the actions in $(\cup_{\ell = 0}^{K-1} A^{'\ell}_i)_{i \in N}$ will be eliminated (i.e., played with vanishing probability) when every player $i \in \hat{N}$ implements projected gradient ascent with appropriately chosen step sizes.

The question is then whether the functions $(h^\ell)_{0 \leq \ell < K}$ can be assembled into $h$ in that manner.
Recall that in Algorithm \ref{alg:ievpp-h}, each $h^\ell$ satisfies the conditions in Definition \ref{eq:epp}. 
We show that if each $h^\ell$ also extends appropriately to the set of all mixed-strategy profiles, then we may assemble them into a global potential function. 

\begin{theorem}\label{thm:iterative}
    Suppose that in Algorithm \ref{alg:ievpp-h}, each $h^\ell : \times_{i \in N} \Delta(A^\ell_i) \rightarrow \mathbb{R}$ 
    extends to a tangential, Lipschitz-differentiable function $\hat{h}^\ell : \times_{i \in N} \Delta(A_i) \rightarrow \mathbb{R}$ satisfying $\partial \hat{h}^\ell(x)/\partial x_i(a_i) = 0$ for any $i \in N, a_i \in A_i \setminus A_i^\ell$.
    Then, there exist strictly positive
    $\epsilon(0), \epsilon(1), ..., \epsilon(K-1)$ such that
    $$ \forall \ x \in \times_{i \in N} \Delta(A_i), \quad \sum_{\ell = 0}^{K-1}\sum_{i \in N} \bilin{\epsilon(\ell) \nabla_i \hat{h}^\ell(x)}{\nabla_i u_i(x)} ~ \geq ~ \sum_{a \in A} \bigg(  \prod_{j \in N} x_j(a_j)\bigg) \mathbb{I}[\exists \ i \in N, a_i \notin A^K_i].$$
    As a consequence, as $T \rightarrow \infty$, if all players implement online gradient ascent with suitably decreasing step sizes, in time-average, outcomes $a$ with some $a_i \in A_i \setminus A^K_i$ will occur with vanishing probability.
\end{theorem}

\begin{remark}\label{remark:anti-pendantic-reviewer}
    For each $A'_i \subseteq A_i$, we have encoded $\Delta(A'_i)$ as a subset of $\mathbb{R}^{A'_i}$, whereas $\Delta(A_i)$ is a subset of $\mathbb{R}^{A_i}$. These are technically two distinct sets, whereas the extension of a function $h : S \rightarrow \mathbb{R}$ to a function $\hat{h} : T \rightarrow \mathbb{R}$ is usually defined when $T \supseteq S$. However, there is an obvious manner in which we consider $\Delta(A'_i)$ as a subset of $\Delta(A_i)$, by considering the isomorphism  
    $$ \phi : \Delta(A'_i) \rightarrow  \{ x_i \in \mathbb{R}^{A_i} \ | \ \ones^T x_i = 1, x_i \geq 0, \text{ and } \  \forall a_i \in A_i\setminus A'_i,\  x_i(a_i) = 0 \},$$
    which maps each $x_i \in \Delta(A'_i)$ to $\phi(x_i) \in \Delta(A_i)$ such that $\phi(x_i)(a_i) = x_i(a_i)$ for $a_i \in A'_i$ and $\phi(x_i)(a_i) = 0$ otherwise. We similarly consider $\times_{i \in N} \Delta(A'_i)$ as a subset of $\times_{i \in N} \Delta(A_i)$.
\end{remark}

To sketch the argument, (\ref{cond:lyapunov}) states that to eliminate actions $(A^{'\ell}_i)_{i \in N}$, we need to find a differentiable function $h^\ell : \times_{i \in N} \Delta(A_i) \rightarrow \mathbb{R}$ such that, restricting attention to the set $\times_{i \in N} \Delta(A^\ell_i)$, its rate of change $\sum_{i \in N} \bilin{\nabla_i h^\ell(\cdot)}{\nabla_i u_i(\cdot)} : \times_{i \in N} \Delta(A^\ell_i) \rightarrow \mathbb{R}$ is minorised (lower-bounded) by a multilinear function which is strictly positive on any mixed-strategy profile $x$ that assigns positive probability to an action in $A^{'\ell}_i$. The technical assumption ensures that $\sum_{i \in N} \bilin{\nabla_i h^\ell(\cdot)}{\nabla_i u_i(\cdot)} : \times_{i \in N} \Delta(A^\ell_i) \rightarrow \mathbb{R}$ also extends to $\sum_{i \in N} \langle{\nabla_i \hat{h}^\ell(\cdot)}{\nabla_i u_i(\cdot)}\rangle : \times_{i \in N} \Delta(A_i) \rightarrow \mathbb{R}$.
If $\hat{h}^\ell$ is twice continuously differentiable, then $\sum_{i \in N} \langle{\nabla_i \hat{h}^\ell(\cdot)}{\nabla_i u_i(\cdot)}\rangle$ is continuously differentiable on $\times_{i \in N} \Delta(A_i)$; moreover, its derivative is bounded on its compact domain.  Therefore, we can extend the non-negative valued multilinear minorant of $\sum_{i \in N} \bilin{\nabla_i h^\ell(\cdot)}{\nabla_i u_i(\cdot)}$ on $\times_{i \in N} \Delta(A^\ell_i)$ to a (potentially negative valued in the interior) multilinear minorant of $\sum_{i \in N} \langle{\nabla_i \hat{h}^\ell(\cdot)}{\nabla_i u_i(\cdot)}\rangle$ on all of $\times_{i \in N} \Delta(A_i)$. The assembly of $(\hat{h}^\ell)_{1 \leq \ell < K}$ into a Lyapunov function then follows from inductive arguments, and the time-average guarantee follows from Proposition \ref{prop:time-avg-guarantees}. 
We formalise this argument in Appendix \ref{sec:proof-thm-iterative}. 

\section{Equilibrium Refinement Beyond Linear Strategy Modifications via Cubic Potentials}\label{sec:cubic}

With Algorithm \ref{alg:ievpp-h}, we have access to an iterative argument to establish time-average convergence to outcomes in which each player $i$ only uses actions in $A_i'\subseteq A_i$.
However, this argument requires actually finding Lipschitz-differentiable \& tangential potential functions $h^\ell$ satisfying \eqref{cond:lyapunov} in order to eliminate actions $(A^{'\ell}_i)_{i \in N}$ in the subgame at each step $\ell$.
In this section, we address the issue of \emph{where to look for them}.

We obtain one set of canditate potential functions through Proposition \ref{prop:scce-constraints}, which provides a 
refinement over the usual no-external regret guarantee for 
gradian ascent dynamics
in games. In particular, according to \citep{ahunbay2024uniqueness}, the constraints generated by functions $h_i^{S_i}$ in (\ref{eq:good-devs}) are sufficient to show that, in (complete information) first-price auctions with at least three buyers of highest valuation, when all buyers implement online gradient ascent with suitable step sizes,
the dynamics: 
(i) time-average converge to equilibrium 
if weak overbidding is allowed, and
(ii) last-iterate converge to equilibrium 
if buyers only bid less than their valuations.
It also turns out that functions $h_i^{S_i}$ (\ref{eq:good-devs}) are suitable
for use in Algorithm \ref{alg:ievpp-h}, providing one canditate set of potentials to search over for the purpose of iterative elimination of strategies.
\begin{observation}\label{obs:extension}
    For player $i$ and $\emptyset \ne S_i \subsetneqq A^\ell_i \subsetneqq A_i$, consider the function $h^{S_i,\ell}_i : \Delta(A_i^\ell) \rightarrow \mathbb{R}$ where 
    $$ h^{S_i,\ell}_i(x) ~ = ~  -\frac{1}{2 |A_i^\ell \setminus S_i|} \Bigg( 1 - \sum_{a_i \in A_i^\ell \setminus S_i} x_i(a_i)\Bigg)^2 - \frac{1}{2} \sum_{a_i \in S_i} x_i(a_i)^2.$$
    Then for $T_i = S_i \cup (A_i \setminus A_i^\ell)$, the function $h_i^{T_i}$ in (\ref{eq:good-devs}) is a Lipschitz-differentiable \& tangential extension of $h^{S_i,\ell}_i$ on $\Delta(A_i)$ such that the conditions of Theorem \ref{thm:iterative} are satisfied.
\end{observation}

\begin{proof}
    Note that, $A_i \setminus T_i = A_i \setminus (S_i \cup (A_i \setminus A^\ell_i)) = A^\ell_i \setminus S_i$. Therefore, by (\ref{eq:good-devs}),
    \begin{align*}
        h_i^{T_i}(x_i) ~ & = ~ -\frac{1}{2|A_i\setminus (S_i \cup (A_i\setminus A^\ell_i))|} \Bigg( 1- \sum_{a_i \in A_i\setminus (S_i \cup (A_i\setminus A^\ell_i))} x_i(a_i) \Bigg)^2 - \frac{1}{2} \sum_{a_i \in S_i \cup (A_i\setminus A^\ell_i)} x_i(a_i)^2 \\
        & = ~ -\frac{1}{2|A^\ell_i\setminus S_i|} \Bigg( 1- \sum_{a_i \in A^\ell_i \setminus S_i} x_i(a_i) \Bigg)^2 - \frac{1}{2} \sum_{a_i \in S_i} x_i(a_i)^2 - \frac{1}{2} \sum_{a_i \in A_i\setminus A^\ell_i} x_i(a_i)^2.
    \end{align*}
    Lipschitz-differentiability is obvious, and so too is the condition $\partial h^{T_i}(x_i) /\partial x_i(a_i) =0$ for $a_i \in A_i \setminus A^\ell_i$. Tangentiality was proven in \cite{ahunbay2025semicoarse}. 
\end{proof}

However, \citet{ahunbay2025semicoarse} show that the semicoarse equilibrium constraints become insufficient to show even time-average convergence to equilibrium when there are at most two buyers of the highest valuation. We note that, for mixed-extensions of normal-form games, Proposition \ref{prop:scce-constraints} characterises \underline{all} linear strategy modifications against which gradient ascent is guaranteed to incur low regret against, and these strategy modifications are \emph{necessarily} generated by tangential quadratic functions. Therefore, to locate another class of \emph{useful} polynomial functions, we 
will investigate Lipschitz-differentiable tangential polynomial functions of degree greater than $2$.

We shall first note that the obvious generalisation of the approach of \cite{ahunbay2025semicoarse} -- characterising the rays of the cone of tangential quadratic functions -- does \emph{not} provide a simple way to characterise such polynomial functions of even cubic order. To wit, consider a generic cubic polynomial
$$ \phi(x_i) = \frac{1}{3} \sum_{\mu,\nu,\gamma \in A_i} C_{\mu\nu\gamma} x_i(\mu) x_i(\nu) x_i(\gamma) + \frac{1}{2} \sum_{\mu,\nu \in A_i} Q_{\mu\nu} x_i(\mu) x_i(\nu) + \sum_{\mu \in A_i} \ell_\mu x_i(\mu),$$
where without loss of generality\footnote{This assumption is without loss of generality, as when we evaluate the gradient we effectively symmetrise both $C$ and $Q$.}, we may assume that $C$ is a supersymmetric tensor of order $3$ and $Q$ is a symmetric matrix. Through Definition \ref{def:tangency} and the resulting tangency constraints (\ref{cons:probability}), (\ref{cons:tangency}) on $\Delta(A_i)$, we conclude that $\phi$ must satisfy the following constraints:
\begin{align}
    \ones^T \nabla_i \phi(x_i) & = 0 \qquad \forall \ x_i \in \Delta(A_i), \nonumber \\
    \nabla_i \phi(x_i)_{a_i} & \geq 0 \qquad \forall \ x_i \in \Delta(A_i) \text{ s.t. } x_i(a_i) = 0. \label{cons:tangency-2}
\end{align}
The first constraint specifies that the strategy modification\footnote{\citet{ahunbay2025semicoarse} show that if $\phi$ is tangential, there exists $\delta>0$ such that  $x_i + \delta \nabla \phi(x_i) \in \Delta(A_i) \ \forall \ x_i \in \Delta(A_i)$.} $x_i \mapsto x_i + \delta\nabla\phi(x_i)$ never changes the total probability ($1$) assigned to all actions, and the second constraint specifies that no action is ever assigned a negative probability. The constraint (\ref{cons:tangency-2}) necessitates that 
$$ \sum_{\mu,\nu \in A_i} C_{a_i \mu\nu} x_i(\mu) x_i(\nu) + \sum_{\mu \in A_i} Q_{a_i \mu} x_i(\mu)  + \ell_{a_i} \geq 0 \qquad \forall \ x_i \in \Delta(A_i) \text{ s.t. } x_i(a_i) = 0.$$
In particular, if $Q,\ell = 0$, then verifying whether (\ref{cons:tangency-2}) holds is equivalent to checking whether the $(|A_i|-1)$-dimensional square matrix $(C_{a_i \mu \nu})_{\mu,\nu \neq a_i}$ is \emph{copositive} for each $a_i \in A_i$. Checking whether a matrix is copositive, however, is \texttt{co-NP}-hard \cite{Dür10,Burer09}, and thus we do not hope that the conic analysis of \cite{ahunbay2025semicoarse} can be repeated to provide an exact characterisation of all quadratic strategy modifications that online gradient ascent satisfies. 

\citet{Dür10} references ways to provide inner approximations for the copositive cone, and it is possible that these methods may provide means to identify strategy modifications that are ``useful'' for equilibrium analysis. We shall defer this to future work, adopting instead an alternative approach to \emph{construct} tangential cubic functions on $\Delta(A_i)$. We remark that the analysis of the Bertrand competition in \cite{ahunbay2025semicoarse} relies on leveraging the more interesting affine-symmetric strategy modifications generated by functions of the form (\ref{eq:good-devs}). We shall thus consider modifying these functions via a \emph{perturbative expansion} to obtain tangential cubic polynomials which generate quadratic strategy modifications gradient ascent incurs low regret against.

Towards this end, \underline{\emph{to simplify notation, we shall drop the subscript $i$ for the rest of this section}}; let $A$ be the set of actions of player $i$, and let $S \subset A$ be a proper subset of $A$. The sets $A$ and $S$ then define a function $h^S$ as in (\ref{eq:good-devs}), which specifies a linear strategy modification mapping the probability assigned to each action $a \in S$ to the uniform probability distribution on $A \setminus S$. We shall then consider a cubic polynomial of the form,
$$g(x) = h^S(x) (1-c^Tx), \quad \text{where } \ones^T c = 0 \text{ and } c(a) = 0 \ \forall \ a \in S.$$

Here, $c$ is to be understood as ``small'' perturbation vector. We want to find conditions on $c$ which are sufficient to ensure that $g$ is tangential. First, we check whether the conservation of probability holds. Note that by \cite{ahunbay2025semicoarse}, we have $\ones^T \nabla h^S(x) = 0$, and therefore, 
$$ \ones^T \nabla g(x) = \ones^T \nabla h^S(x) (1-c^T x) - h^S(x) \ones^T c = 0.$$

Second, we must check whether $\nabla g$ ever assigns negative probability to any action. Note that  
$$ \nabla g(x)_a = \nabla h^S(x)_a (1-c^T x) - h^S(x) c_a.$$
First, suppose that $a \in S$, then $c_a = 0$ for any $a \in S$ by construction. Moreover, $\nabla h^S(x)_a = -x(a)$, which actually equals zero whenever $x(a) = 0$. Therefore, $\nabla g$ never attempts to assign negative probability to any $a \in S$ as intended. 

Finally, suppose that $a \in A \setminus S$; here, we will need to derive further constraints on the perturbation vector $c$. In this case, writing out the term of the gradient explicitly,
$$ \frac{\partial g(x)}{ \partial x(a) } = \frac{1}{|A\setminus S|} \left( \sum_{a' \in S} x(a') \right) \left(1 - \sum_{a' \notin S} c(a') x(a')\right) - c(a) h^S(x).$$
Notice first, that if all $c(a')$ are $\leq 1$ and if $c(a) \geq 0$, then since $h$ is negative semidefinite, $\partial g(x) / \partial x(a) \geq 0$. Since we are looking for sufficient conditions for $g$ to be tangential, we shall enforce $c(a') \leq 1$ for any $a' \in A$. It then remains to consider the case $c(a) < 0$. 
Recalling the form of $h^S(x)$ through (\ref{eq:good-devs}), we find the lower bound 
\begin{align*}
    h^S(x) & = - \frac{1}{2 |A\setminus S|} \left(1-\sum_{a \notin S} x(a) \right)^2 - \frac{1}{2} \sum_{a \in S} x(a)^2 = -\frac{1}{2|A\setminus S|} \left(\sum_{a \in S} x(a) \right)^2 - \frac{1}{2} \sum_{a \in S} x(a)^2 \tag{$x\in \Delta(A)$} \\
    & \geq -\frac{1}{2} \left( 1 + \frac{1}{|A\setminus S|}\right) \left(\sum_{a \in S} x(a) \right)^2. \tag{$x \geq 0 \Rightarrow (\sum_ax(a))^2\geq \sum_a x(a)^2$}
\end{align*}
Therefore, setting $\delta = \sum_{a' \in S} x(a')$ and $\bar{c} = \max_{a'} c(a')$, we see that 
\begin{align}
    \frac{\partial g(x)}{ \partial x(a) } & \geq \frac{1}{|A\setminus S|} \delta \left(1 - \bar{c} (1-\delta)\right) + \frac{1}{2} c(a) \left( 1 + \frac{1}{|A\setminus S|}\right) \delta^2 \nonumber\\
    & = \delta^2 \left( \frac{\bar{c}}{|A\setminus S|} + \frac{1}{2} c(a) \left( 1 + \frac{1}{|A\setminus S|}\right)\right) + \frac{1}{|A\setminus S|}(1-\bar{c}) \delta.\label{eq:lb-tangency}
\end{align}
In this case, if the coefficient of $\delta^2$ in (\ref{eq:lb-tangency}) is $\geq 0$, since $\bar{c} \leq 1$, the right-hand size is strictly increasing in $\delta$ and the minimum of (\ref{eq:lb-tangency}) is attained when $\delta = 0$. In this case, $\partial g(x) / \partial x(a) \geq 0$ as expected. In turn, if the coefficient of $\delta^2$ is negative, then we have a concave minimisation problem, implying that the minimum is attained at $\delta \in \{0,1\}$. The case $\delta = 0$ has (\ref{eq:lb-tangency}) $= 0$, so we fix $\delta = 1$, which yields 
$$ \frac{1}{2} c(a) \left(1 + \frac{1}{|A\setminus S|}\right) + \frac{1}{|A\setminus S|} \geq 0 \Rightarrow c(a) \geq -\frac{2}{|A\setminus S| + 1}.$$

\begin{theorem}\label{thm:cubic-candidates}
    Let $A$ be a finite set, and $S$ a proper subset of $A$. Suppose that $c$ is a vector satisfying \textnormal{(1)} $c(a) = 0$ for any $a \in S$, \textnormal{(2)} $\sum_{a \in A} c(a) = 0$, and \textnormal{(3)} $c(a) \in [-2/(|A\setminus S|+1),1]$ for any $a \in A \setminus S$. Then the function 
    $g^{S,c} : \Delta(A) \rightarrow \mathbb{R}$, defined as $g^{S,c}(x) = h^S(x)(1-c^Tx)$, is tangential over $\Delta(A)$. As a consequence, against any sequence of payoff gradients $p^t$ such that $\|p^t\| \leq G$, online gradient ascent over $\Delta(A)$ satisfies 
    $$ \sum_{t = 1}^T \bilin{\nabla g^{S,c}(x)}{p^t} ~ \leq ~ \frac{4(2+\|c\|)}{T\eta_T}  + \frac{\sum_{t=1}^T \eta_{it}}{T} \cdot 4(1+2\|c\|)G^2,$$
\end{theorem}

\begin{proof}
    We proved already the tangentiality of $g^{S,c}$, so it remains to show the bound. We invoke Proposition \ref{prop:reg-guarantee-gd} here, noting that $g^{S,c}(x) \in [-2,0]$, and that $\nabla^2 g^{S,c}(x) = \nabla^2 h(x) (1-c^T x) - \nabla h(x) c^T - c \nabla h(x)^T$. The Lipschitz modulus of $h$ is bounded above by $1$, $|c^T x| \leq 1$, and $\|\nabla h\| \leq 2$, which implies that the Lipschitz modulus of $g^{S,c}$ is bounded by $2 + 4\|c\|$. Meanwhile, the magnitude of $\nabla g^{S,c}$ is bounded, $\|\nabla g^{S,c}(x) \| \leq \|\nabla h^S(x)\||1-c^T x| + |h(x)|\|c\| \leq 4 + 2\|c\|$.
\end{proof}

To conclude this section, we will point out two properties of the strategy modifications prescribed by $g^{S,c}$. First, we remark that in a normal-form game, we may write the bilinear term $\langle{\nabla_i g^{S_i,c}_i(x)},{\nabla_i u_i(x)}\rangle$ in the form of its deviations, as was done for functions $h^{S_i}_i$ in (\ref{eq:deviation-form}). The proof of this can be found in Appendix \ref{sec:proof-deviation-form-cubic}.

\begin{lemma}\label{lem:deviation-form-cubic}
    In the mixed-extension of a normal-form game, for each mixed-strategy profile $x \in X$ such that $\sum_{a_i \in S_i} x_i(a_i) > 0$, the rate of change $\langle{\nabla_i g^{S_i,c}_i(x_i)},{\nabla_i u_i(x)}\rangle$ equals 
    \begin{equation}\label{eq:deviation-form-cubic}
        \sum_{a_{-i} \in A_{-i}} \sum_{a_i \in S_i} \left( \prod_{j \in N} x_j(a_j) \right) \left[\sum_{a'_i\in A_i\setminus S_i} u_i(a'_i,a_{-i}) \left( \frac{1-c^T x_i}{|A_i\setminus S_i|} - \frac{c(a'_i) h_i^{S_i}(x_i)}{\sum_{a''_i \in S_i} x_i(a''_i)}  \right) -u_i(a)\left(1-c^T x_i\right)\right].
    \end{equation}
\end{lemma}

Thus, unlike the form of deviations (\ref{eq:deviation-form}) for $h_i^{S_i}$, in this case the strategy modifications are non-linear in the present mixed-strategy used by the player. For any $a_i \in S_i$, $\partial g^{S_i,c}_i(x_i) / \partial x_i(a_i) = -x_i(a_i) (1-c^Tx_i)$, and hence as in the case of semicoarse equilibria, the strategy modification attempts to deviate away from each $a_i \in S_i$ at equal rate proportional to $1-c^Tx_i$. However, if $a_i \notin S_i$, then 
$\partial g^{S_i,c}_i(x) / \partial x_i(a_i) = (1/{|A_i\setminus S_i|}) \delta (1-c^T x_i) - c(a_i) h^{S_i}_i(x_i)$,
where $\delta$ is the probability assigned to actions in $S_i$, and $h^{S_i}_i(x) = \Theta(\delta^2)$. Thus, we are allowed to deviate from actions in $S_i$ to \emph{non-uniform} distributions over $A_i\setminus S_i$, but the magnitude of these non-uniformities are bounded linearly with the amount of probability assigned to actions in $S_i$ by $x_i$. 

Second, we note that the cubic potentials $g^{S_i,c}_i$ satisfy the conditions of Theorem \ref{thm:iterative}, meaning they are suitable for application in interative elimination of strategies via potential proofs. We defer the proof to Appendix \ref{sec:proof-cubic-extensible}.

\begin{proposition}\label{prop:extension}
    For player $i$ and $S_i \subset A^\ell_i \subset A_i$, consider the function $g^{S_i,c,\ell}_i : \Delta(A^\ell_i) \rightarrow \mathbb{R}$, defined 
    $$ g^{S_i,c,\ell}_i(x_i) = (1-c^Tx_i)\left[ -\frac{1}{2 |A_i^\ell \setminus S_i|} \left( 1 - \sum_{a_i \in A_i^\ell \setminus S_i} x_i(a_i)\right)^2 - \frac{1}{2} \sum_{a_i \in S_i} x_i(a_i)^2 \right] \ \forall \ x_i \in \Delta(A_i^\ell).$$
    Let $T_i = S_i \cup (A_i \setminus A_i^\ell)$. Let $d \in \mathbb{R}^{A_i}$ be such that $d(a_i) = c(a_i)$ for any $a_i \in A^\ell_i$ and $0$ otherwise. Then the function $g_i^{T_i,d}$ is a Lipschitz-differentiable \& tangential extension of $h^{S_i,c,\ell}_i$ on $\Delta(A_i)$.
\end{proposition}

\section{Analysis of Gradient Dynamics in First-Price Auction}
\label{sec:FPA}
We shall demonstrate the expressivity of our equilibrium refinement through the analysis of the first-price auction with complete information. We will use the discretised model of \cite{deng2022nash}, 
where $N$ buyers have positive integer valuations sorted in decreasing order $v_1 \ge v_2 \ge \cdots \ge v_N$. 
Each buyer $i \in N$ has an action (bid) set $B_i = \{0, 1, ..., v_i - 1\}$, which implicitly assumes no (weak) overbidding. Given a vector of bids $b \in B$, 
the utility of buyer $i$ is 
$$ u_i(b) = \begin{cases}
    \frac{v_i - b_i}{|\arg \max_{j \in N} b_j |} & i \in \arg \max_{j \in N} b_j, \\
    0 & \text{otherwise.}
\end{cases}$$
Our result is a proof that we may eliminate any bid $b_1 \not\simeq v_2$ of buyer $1$ via Algorithm \ref{alg:ievpp-h} (IEPP).
\begin{theorem}\label{thm:first-price}
    In our model of a first-price auction with complete information, suppose there are $N \geq 2$ buyers, and $v_1 \geq 2$. Then the following subset of bids $B'_1$ of buyer $1$ may be eliminated via an additively separable potential $h(x) = h_1(x_1)+h_2(x_2)$, where
    \begin{enumerate}
        \item if $v_1 = v_2, v_2+1$, then $B'_1 = \{b_1 \in B_1 \ | \ b_1 < v_2 -2\}$,
        \item else if $v_1 > v_2 + 1$, then $B'_1 = \{b_1 \in B_1 \ | \ b_1 \neq v_2\}$.
    \end{enumerate}
    As a consequence, if buyers $1$ and $2$ implement online gradient ascent for $T$ periods with non-increasing step sizes satisfying $1/\eta_{iT}, \sum_{t = 1}^T \eta_{it} = o(T)$, in time-average buyer $1$ chooses a bid $\notin \{v_2-2,v_2-1,v_2\}$ with probability $o(1)$.
\end{theorem}

Our proof strategy will be as follows.  We iteratively eliminate the lowest bid of buyer $1$, implementing Algorithm \ref{alg:ievpp-h}. At iteration $\ell \geq 0$, we thus have bid set $B_1^\ell = \{\ell,\ell+1,...,v_1-1\}$ for buyer $1$, and $B_i^\ell = B_i$ for any $i > 1$. Iteratively eliminating all bids $< v_2 -2$ then shows (1). To prove (2) and (3), we then need to analyse the resulting subgame. Proposition \ref{prop:first-elim} shows the form of the potentials we use to eliminate bids $0 \leq \ell < v_2 -2$, whereas the rest of the proof of Theorem \ref{thm:first-price} is deferred to Appendix \ref{appendix:first-price}. 

\begin{remark}\label{remark:edge-case}
    The result of Theorem \ref{thm:first-price} is in fact tight when either $v_1 = v_2$ or $v_1 \geq v_2+2$. When $v_1 = v_2 > v_3+2$, then there is an pure strategy Nash equilibrium where both buyers $1$ \& $2$ bid $v_1-2$, as well as another Nash equilibrium where they both bid $v_1-1$, thus neither bid may be eliminated. Meanwhile, through the case $v_1 \geq v_2+2$, we see that if buyers do not weakly overbid, all Nash equilibria must have buyer $1$ bidding exactly $v_2$. However, the case $v_1 = v_2 +1$ exhibits an interesting criticality phenomenon; the bid $v_2-2$ can also be eliminated for buyer $1$, but how to do so is contingent on whether we make further assumptions on buyer $3$ (c.f. \ref{appendix:edge-case}).
\end{remark}

\begin{proposition}\label{prop:first-elim}
    At iteration $\ell < v_2 -2$, buyer $1$'s bid $\ell$ may be eliminated via a potential proof by $h^\ell(x) = \alpha^\ell h_1^\ell(x_1)+\beta ^\ell g_2^\ell(x_2)$, where the functions $h_1^\ell = h_1^{\{\ell\},\ell}, g_2^\ell = g_2^{\{0,1,..,\ell\},c^\ell,\ell}$ are respectively defined as in Observation \ref{obs:extension} and Proposition \ref{prop:extension}, and we fix 
    \begin{align*}
        \alpha^\ell = \frac{2N(v_1-\ell-1)}{v_1-v_2+1}, ~ \beta^\ell = \max\left\{2\alpha^\ell +\frac{8}{v_2-\ell}, 32\right\}, ~ \text{and } c^\ell(b_2) = \begin{cases}
            0 & b_2 \leq \ell, \\
            1/2 & b_2 = \ell+1, \\
            -1/2(v_2 - \ell - 2) & b_2 > \ell+1.
        \end{cases}
    \end{align*}
\end{proposition}

\begin{proof}
    For any $x_2 \in \Delta(B_2)$, we shall denote $\phi_2^\ell(x_2) = \sum_{b_2 = 0}^{\ell-1} x_2(b_2)$ and $\psi_2^\ell(x_2) = \sum_{b_2 = \ell+1}^{v_2-1} x_2(b_2)$; that is, they respectively correspond to the probability that buyer $2$ bids strictly lesser or strictly greater than $\ell$. Now, we first note that by \eqref{eq:deviation-form}, for any $x \in \Delta(B^\ell)$, $\bilin{\nabla_1 h_1^{\ell}(x_1)}{\nabla_1 u_1(x)}$ is a multilinear function, with value 
    \begin{align*}\sum_{b_{-1} \in B_{-1}} x_1(\ell) \left( \prod_{j \neq 1} x_j (b_j) \right) \left[ \sum_{b'_1 = \ell+1}^{v_1-1} \frac{u_1(b_1',b_{-1})}{v_1 -\ell-1} - u_1(\ell,b_{-1}) \right].\end{align*}
    Second, we shall denote $g_2^\ell(x_2) = g_2^{\{0,1,...,\ell\},c^\ell,\ell}(x_2)$, the latter as in Proposition \ref{prop:extension}. Then, by \eqref{eq:deviation-form-cubic}, $\bilin{\nabla_2 g_2^\ell(x_2)}{\nabla_2 u_2(x)}$ is multilinear over all $x_i$ \emph{except} for $x_2$, where
    \begin{align*}
        & \bilin{\nabla_2 g_2^\ell(x_2)}{\nabla_2 u_2(x)} = \\ & \sum_{b_{-2} \in B_{-2}} \sum_{b_2 = 0}^\ell \left(\prod_{j \in N} x_j(b_j) \right)   \left[\sum_{b'_2 = \ell+1}^{v_2 -1 } u_2(b_2',b_{-2}) \left( \frac{1-\bilin{c^\ell}{x_2}}{v_2-\ell-1} - \frac{c^\ell(b'_2) h_2^\ell(x_2)}{\phi_2^\ell(x_2) + x_2(\ell)}  \right) -u_2(b)\left(1-\bilin{c^\ell}{x_2}\right)\right].
    \end{align*}
    Meanwhile by (\ref{cond:lyapunov}), we would like to show that for any $x \in \times_{i \in N} \Delta(B_i^\ell)$,
    \begin{equation}\label{eq:desired}\alpha^\ell \bilin{\nabla_1 h_1^\ell(x_1)}{\nabla_1 u_1(x)} + \beta^\ell \bilin{\nabla_2 g_2^\ell(x_2)}{\nabla_2 u_2(x)} \geq \sum_{b \in B^\ell} \left( \prod_{j \in N} x_j(b_j) \right) \mathbb{I}[b_1 = \ell] = x_1(\ell).\end{equation}
    However, the right-hand side of \eqref{eq:desired} is also a multilinear function. Therefore, \emph{by multilinearity of both sides of \eqref{eq:desired}, it is sufficient to show the inequality whenever $x_2 \in \Delta(B_2)$ is fixed arbitrarily, and players $\neq 2$ are prescribed pure strategies $b_{-2}$ via $x_{-2}$.} We proceed by case analysis on all such bid profiles $b_{-2}$ for buyers other than $2$. 

    First, assume that the buyer $1$ bids $b_1 > \ell$. Then buyer $1$ has no strategy modification prescribed to them by $\nabla_1 h_1$, whereas buyer $2$'s strategy modification deviates from bids $\leq \ell$ — therefore being weakly utility improving. On the other hand, the right-hand side of (\ref{eq:desired}) equals zero whenever $b_1 > \ell$. Therefore, (\ref{eq:desired}) holds in this case. 
    
    Henceforth, we may restrict attention to the case in which buyer $1$ bids $b_1 = \ell$. Now suppose that some buyer $j > 2$ uses a bid $b_j > \ell$. Then note that both buyer $1$ and buyer $2$ deviate only from bids $\leq \ell$, so both strategy modifications result in non-negative changes in utility. Moreover, if buyer $1$ bids $\ell$, then their utility is guaranteed to be zero, whereas when they deviate to the uniform distribution on $\{\ell+1,\ell+2,...,v_1-1\}$, they are guaranteed a utility of at least $(v_1-v_2+1)/N(v_1-\ell-1)$ -- this covers the possibility that $v_i = v_2$ for any buyer $i \geq 2$, and all buyers $\neq 1$ bid $v_2-1$ with probability $1$. Therefore, if $\alpha^\ell \geq (v_1-\ell-1)N/(v_1 - v_2 + 1)$, the desired inequality holds.

    We can thus turn attention to the case when each buyer $\geq 3$ bids an amount $\leq \ell$. As a subcase, suppose that at least one buyer $i \geq 3$ actually bids $\ell$. Denote by $\omega = | \arg \max_{j \neq 2} b_j|+1$, and note that $\omega \geq 2$ since $b_1 = \ell$. We bound the change in utility for buyer $2$, $\bilin{\nabla_2g_2^\ell(x_2)}{\nabla_2u_2(x)}$, which equals
    \begin{equation}
        \sum_{b_2 = 0}^\ell x_2(b_2) \cdot  \left[\sum_{b'_2 = \ell+1}^{v_2 -1 } u_2(b_2',b_{-2}) \left( \frac{\left(1-\bilin{c^\ell}{x_2}\right)}{v_2-\ell-1} - \frac{c^\ell(b'_2) h_2^\ell(x_2)}{\phi_2^\ell(x_2) + x_2(\ell)}  \right) -u_2(b)\left(1-\bilin{c^\ell}{x_2}\right)\right]. \label{eq:bound-2}
    \end{equation}
    First, note that for any $b_2 \leq \ell$, $u_2(b) \leq (v_2-\ell)/\omega$; if $b_2 = \ell$ the equality holds, whereas $u_2(b) = 0$ whenever $b_2 < \ell$. Meanwhile, $|c(b_2)| \leq 1/2$ for any $b_2$. Therefore, if $b_2 \leq \ell$, then
    \begin{align}
        \left(1-\bilin{c^\ell}{x_2}\right) \left( \sum_{b_2' = \ell+1}^{v_2 - 1} \frac{u_2(b'_2,b_{-2})}{v_2-\ell-1} - u_2(b) \right) & \geq  \left(1-\bilin{c^\ell}{x_2}\right) \left( \sum_{b_2' = \ell+1}^{v_2 - 1} \frac{v_2 - b'_2}{v_2-\ell-1} - \frac{v_2 - \ell}{\omega} \mathbb{I}[b_2 = \ell] \right) \nonumber \\
        = \left(1-\bilin{c^\ell}{x_2}\right) \cdot \left( \frac{1}{2} - \frac{\mathbb{I}[b_2 = \ell]}{\omega} \right) (v_2-\ell) & \geq \frac{1}{2} \left( \frac{1}{2} - \frac{\mathbb{I}[b_2 = \ell]}{\omega} \right) (v_2 - \ell). \label{eq:bound-2-scce}
    \end{align}
    In turn, we note that $h_2^\ell(x_2) \leq -\frac{(\phi^\ell_2(x_2) + x_2(\ell))^2}{2(v_2-\ell-1)}$, which implies that
    \begin{align}
        \sum_{b_2' = \ell+1}^{v_2 - 1} u_2(b_2',b_{-2})\cdot\frac{-c^\ell(b_2) h^\ell_2(x_2)}{\phi^\ell_2(x_2) + x_2(\ell)} & = \frac{|h_2^\ell(x_2)|}{\phi^\ell_2(x_2) + x_2(\ell)} \cdot \left( \frac{v_2 - \ell - 1}{2} -\sum_{b'_2 = \ell + 2}^{v_2 - 1} \frac{v_2-b'_2}{2(v_2 - \ell - 2)}\right) \nonumber\\
        & \geq \frac{\phi^\ell_2(x_2) + x_2(\ell)}{2(v_2-\ell-1)} \cdot \left( \frac{v_2 - \ell -1}{4}\right) = \frac{\phi^\ell_2(x_2) + x_2(\ell)}{8}. \label{eq:bound-2-perturb}
    \end{align}
    Therefore, bounding \eqref{eq:bound-2} as $\sum_{b_2 = 0}^\ell x_2(b_2) \cdot [$\eqref{eq:bound-2-scce} $+$ \eqref{eq:bound-2-perturb}$]$, we see that
    \begin{equation}\label{eq:bound-2-final}
        \bilin{\nabla_2g_2^\ell(x_2)}{\nabla_2u_2(x)} \geq \frac{\phi_2^\ell(x_2) (v_2-\ell)}{4} + \frac{(\phi_2^\ell(x_2)+x_2(\ell))^2}{8}.
    \end{equation}
    Meanwhile, note that buyer $1$'s payoff change $\bilin{\nabla_1 h_1^\ell(x_1)}{\nabla_1 u_1(x)}$ can be decomposed into three parts, conditional on the bid of buyer $2$, which we lower bound:
    \begin{enumerate}
        \item If buyer $2$ bids an amount $< \ell$, utility goes down from $(v_1 - \ell)/(\omega-1)$ to $(v_1-\ell)/2$. In expectation, this is a change of $-\phi_2^\ell(x_2) (v_1-\ell)(1/(\omega-1) - 1/2) \geq -\phi_2^\ell(x_2) (v_1-\ell)/2$.
        \item If buyer $2$ bids exactly equal to $\ell$, through a similar argument, we have an in expectation change of utility of $x_2(b_2) (v_1-\ell) (1/\omega - 1/2) \geq 0$.
        \item If buyer $2$ bids $> \ell$, utility goes from $0$ to at least $(v_1 - v_2 + 1)/2(v_1-\ell-1)$. In expectation, this is a lower bound of $\psi_2 ^\ell(x_2) (v_1 - v_2 + 1)/2(v_1-\ell-1)$, attained when $x_2(v_2-1) = \psi_2^\ell(x_2)$.
    \end{enumerate}
    Or explicitly,
    \begin{equation}\label{eq:bound-1-final}
        \bilin{\nabla_1 h_1^\ell(x_1)}{\nabla_1 u_1(x)} \geq \frac{\psi_2^\ell(x_2)(v_1-v_2+1)}{2(v_1-\ell-1)}-\frac{\phi_2^\ell(x_2) (v_1-\ell)}{2}.
    \end{equation}
    
    This allows us to piece together our lower bound for \eqref{eq:desired} via $\alpha^\ell \eqref{eq:bound-1-final} + \beta^\ell \eqref{eq:bound-2-final}$,
    \begin{align*}
        & \alpha^\ell \bilin{\nabla_1 h_1^\ell(x_1)}{\nabla_1 u_1(x)} + \beta^\ell \bilin{\nabla_2 g_2^\ell(x_2)}{\nabla_2 u_2(x)} \\  \geq \ &  \frac{\alpha^\ell\psi_2^\ell(x_2)(v_1 - v_2 + 1)}{2(v_1-\ell-1)} + \phi_2^\ell(x_2) (\beta^\ell - 2\alpha^\ell)\left( \frac{v_2-\ell}{4} \right) + \frac{\beta^\ell(\phi_2^\ell(x_2) + x_2(\ell))^2}{8} \\
        \geq \ & 2 \psi_2^\ell(x_2) + 2\phi_2^\ell(x_2) + x_2(\ell)^2 \left(\frac{\beta^\ell}{8}\right),
    \end{align*}
    where the last inequality holds by our choices of $\alpha^\ell = 2N(v_1-\ell-1)/(v_1-v_2+1)$, $\beta^\ell \geq 2\alpha^\ell + 8/(v_2-\ell)$ in the statement of the proposition, as well as noting that  $(\phi_2^\ell(x_2) + x_2(\ell))^2 \geq x_2(\ell)^2$. To see that \eqref{eq:desired} is lower bounded by $1$, note that if $\phi_2^\ell(x_2) + \psi_2^\ell(x_2) \geq 1/2$, then we are good. Else, $x_2(b_2) \geq 1/2$, in which case \eqref{eq:desired} holds by our choice of $\beta^\ell \geq 32$.
\end{proof}

\section{Conclusions \& Open Problems}

All-in-all, our analysis provides a clear proof of time-average revenue and welfare optimality (up to a constant) of repeated first-price auctions with complete information, when the two buyers with the highest valuations both learn using online gradient ascent. This was made possible through identifying a set of quadratic strategy modifications against which online gradient ascent incurs low regret, which along with the semicoarse equilibrium property (Proposition \ref{prop:scce-constraints}.(1)) allowed us to iteratively eliminate the lowest bid of the highest valuation buyer until we were left with an easy-to-analyse game.

Several questions remain open. First is the matter of \emph{bound refinement}. Theorem \ref{thm:first-price} suggests that, if buyers $1$ and $2$ both use step sizes satisfying $\eta_T,\sum_{t = 1}^T \eta_t \propto o(T)$, then in time-average buyer $1$ bids an amount $\notin \{v_2-2,v_2-1,v_2\}$ with probability $C \cdot (\eta_T+\sum_{t=1}^T \eta_t)$ for some constant $C$. The existence of this constant is inferred through the black box application of Theorem \ref{thm:iterative}. However, the proof of Theorem \ref{thm:iterative} results in a potential whose gradient has magnitude and Lipschitz modulus both exponential in $v_2$. We suspect that for general games, exponential time to convergence might be necessary. However, it may be possible to optimise over the coefficients to each $h_1^\ell,g_2^\ell$ in Proposition \ref{prop:first-elim} to obtain a time-average guarantee polynomial in $v_1$ and $v_2$. 

Second, our framework of analysis is valid for any normal-form game, and we suspect there are other settings in which it is applicable. An obvious candidate is other auctions and contests, especially the case of the Bayesian first-price auction with i.i.d.~buyer valuations, which was left open in \cite{ahunbay2024uniqueness}. We also wonder whether the semicoarse equilibrium property along with our cubic potentials are sufficient to demonstrate non-manipulability of gradient ascent in the setting of \citep{braverman2018selling}.

It is also interesting when our cubic potentials are insufficient to capture convergence behaviour. Are there other polynomial potentials which are \emph{useful}? We remark that our work does not exhaust the set of all tangential cubic functions on the probability simplex. A simple description of this set should not be possible under standard complexity theoretic assumptions, due to the aforementioned \texttt{co-NP}-hardness of separating over the set of copositive matrices. However, approximation schemes involving semidefinite programming \citep{Dür10,Parrilo00,Laserre00} might allow us to numerically investigate whether useful superquadratic polynomials exist. We remark that \citep{souaiby2021lyapunov} provides one such scheme for Lyapunov function approximation.

A final question is whether our methods of analysis can be extended to other gradient-based algorithms, such as Hedge \citep{freund1997decision}. \citet{cai2025new} extend the analysis of local equilibrium also to online mirror ascent, and \citet{ahunbay2024first} notes that if the regulariser is steep, these guarantees can be stated in first-order. However, up until now there have not been any explicit results identifying classes of functions which generate strategy modifications useful for proving time-average guarantees in this setting. Whether the investigation of appropriate classes of deviation generating functions would yield a more refined analysis of regularised learning is thus an open problem.

\section*{Acknowledgments}
The work of Mete \c{S}eref Ahunbay was funded by a Walter Benjamin Fellowship from the German Research Foundation (\emph{Deutsche Forschungsgemeinschaft}), project nr.~558988238.

\bibliographystyle{ACM-Reference-Format}
\bibliography{sample-bibliography}

\newpage
\appendix

\section{Smooth Games \& Mixed-Extensions}\label{sec:smooth-statement}

The guarantees for online gradient ascent provided in \cite{ahunbay2024first,ahunbay2025semicoarse,cai2025new} in fact apply to the setting of a \emph{smooth game}. In this appendix, for reference, we shall provide how their setting relates to mixed-extensions of normal-form games, retrieving Proposition \ref{prop:reg-guarantee-gd}. Towards this end, we first provide a definition of a smooth game; the main distinction between a smooth game and a normal-form game is that the action sets are now closed \& convex action sets within some finite dimensional Euclidean space, and the utilities are Lipschitz continuously differentiable with a bound on the magnitude of its gradient.

\begin{definition}
    A \textbf{smooth game} $\Gamma$ is a tuple $(N,(X_i)_{i \in N}, (u_i)_{i \in N})$, where (1) $N$ is the number (and by choice of notation, the set) of players, (2) for each player $i$, $X_i \subseteq \mathbb{R}^{D_i}$ is the set of actions of player $i$, where we denote $X \equiv \times_{i \in N} X_i$ as the set of action profiles or outcomes of the game, and (3) each $u_i : X \rightarrow \mathbb{R}$ is the utility or payoff function of player $i$, satisfying
    \begin{enumerate}[label=\it{(\alph*)}]
        \item (bounded gradients) $\|\nabla_iu_i(x)\| \leq G_i$ for every $x \in X$, and
        \item (bounded Lipschitz modulus) $\|\nabla_i u_i(x) -\nabla_iu_i(y)\| \leq L_i \|x-y\|$ for any $x,y \in X$.
    \end{enumerate}
    We shall denote $D = \sum_{i \in N} D_i$ for the dimension of $X$, $\vec{G} = (G_i)_{i \in N}$, and $\vec{L} = (L_i)_{i \in N}$.
\end{definition}

Then, the mixed-extension of any normal-form game is itself a smooth game; for each player $i$, the probability simplex $\Delta(A_i)$ is compact \& convex, and the utility functions are obviously smooth since they are polynomials over $\mathbb{R}^D$, with 
$ u_i(x) = \sum_{a \in A} \left( \prod_{j \in N} x_j(a_j) \right) u_i(a)$ 
. Moreover, we may recover the following \emph{effective bounds} on the utility gradient.

\begin{lemma}\label{lem:GL-bounds}
    In the mixed-extension of a normal-form game, for each player $i$, we may set 
    \begin{enumerate}
        \item $G_i = D_i \cdot \max_{a_{-i} \in A_{-i}, a_i,a'_i \in A_i} \left\{u_i(a) - u_i(a'_i,a_{-i})\right\}$, and 
        \item $L_i = D \cdot \max_{a \in A} |u_i(a)|$.
    \end{enumerate}
\end{lemma}

\begin{proof}
    Note that each $\Delta(A_i)$ lies on the affine plane $\ones^T x_i=1$, thus, working in the affine span of each $\Delta(A_i)$, it is sufficient to consider the magnitude of the vector $\nabla_i u_i(x) - \frac{1}{D_i} \ones\ones^T\nabla_iu_i(x)$. This vector has components, 
    \begin{align*}\frac{\partial u_i(x)}{\partial x_i(a_i)} & = \sum_{a_{-i} \in A_{-i}} \left( \prod_{j \neq i} x_j(a_j) \right) \left(u_i(a) - \frac{1}{D_i} \sum_{a'_i \in A_i} u_i(a'_i,a_{-i})\right) \\
    & = \frac{1}{D_i} \sum_{a'_i \in A_i} \sum_{a_{-i} \in A_{-i}} \left( \prod_{j \neq i} x_j(a_j) \right) \left(u_i(a) - u_i(a'_i,a_{-i})\right) \\
    & \leq \frac{1}{D_i} \sum_{a'_i \in A_i} \sum_{a_{-i} \in A_{-i}} \left( \prod_{j \neq i} x_j(a_j) \right) \max_{\hat{a}_i,\hat{a}'_i \in A_i, \hat{a}_{-i}\in A_{-i}} \left\{u_i(\hat{a}) -  u_i(\hat{a}'_i,\hat{a}_{-i}) \right\} \\
    & = \max_{\hat{a}_i,\hat{a}'_i \in A_i, \hat{a}_{-i}\in A_{-i}} \left\{u_i(\hat{a}) -  u_i(\hat{a}'_i,\hat{a}_{-i}) \right\},
    \end{align*}
    from which (1) follows. Meanwhile, (2) is essentially the trivial bound; given any $x \in X$, we have 
    \begin{align*}
        \frac{\partial^2 u_i(x)}{\partial x_k(a_k) \partial x_\ell(a_\ell)} = \begin{cases}
            0 & k = \ell, \\
            \sum_{a_{-\{k,\ell\}} \in A_{-\{k,\ell\}}} \left( \prod_{j \neq k,\ell} x_j(a_j)\right) u_i(a) & k \neq \ell.
        \end{cases}
    \end{align*}
    Then, by the Gershgorin circle theorem, a bound on the largest eigenvalue of $\nabla^2 u_i(x)$ is
    \begin{align*}
        \max_{k \in N, a_k \in A_k} \sum_{\ell \in N} \sum_{a_\ell \in A_\ell} \left| \frac{\partial^2 u_i(x)}{\partial x_k(a_k)\partial x_\ell(a_\ell)} \right|  & = \max_{k \in N, a_k \in A_k} \sum_{\ell \in N} \sum_{a_\ell \in A_\ell} \left| \sum_{a_{-\{k,\ell\}} \in A_{-\{k,\ell\}}} \left( \prod_{j \neq k,\ell} x_j(a_j)\right) u_i(a) \right| \\
        & \leq \max_{k \in N, a_k \in A_k} \sum_{\ell \in N} \sum_{a_\ell \in A_\ell} \sum_{a_{-\{k,\ell\}} \in A_{-\{k,\ell\}}} \left( \prod_{j \neq k,\ell} x_j(a_j)\right)  \left| u_i(a) \right| \\
        & \leq \max_{k \in N, a_k \in A_k} \sum_{\ell \in N} D_\ell \cdot \max_{a \in A} |u_i(a)| = D \cdot \max_{a \in A} |u_i(a)|.
    \end{align*}
\end{proof}

Now, restricting attention to Lipschitz continuous vector fields, the definition of an $\epsilon$-local correlated equilibrium in \cite{ahunbay2024first} is equivalent to 

\begin{definition}
    A distribution $\sigma \in \Delta(X)$ is called an $\epsilon$\textbf{-local correlated equilibrium} with respect to a set $F$ of magnitude bounded, Lipschitz continuous vector fields $f : X \rightarrow \mathbb{R}^D$ if for every $f \in F$,
    $$ \sum_{i \in N} \mathbb{E}_{x \sim \sigma}\left[\bilin{\proj{\TC{X_i}{x_i}}{f_i(x)}}{\nabla_i u_i(x)} \right] \leq \epsilon \cdot \poly(\vec{G},\vec{L},G_f,L_f),$$
    where $G_f$ and $L_f$ are respectively the bounds on the magnitude and the Lipschitz modulus of $f$, i.e. $\|f(x)\| \leq G_f$ for any $x \in X$ and $\|f(x)-f(y)\| \leq L_f \|x-y\|$ for any $x,y \in X$. 
\end{definition}

By Lemma \ref{lem:GL-bounds}, we see that a polynomial over $(\vec{G},\vec{L},G_f,L_f)$ is necessarily bounded above by a polynomial over $(\vec{u},G_f,L_f)$, and thus Definition \ref{def:lcce} is an appropriate adaptation of the equilibrium concept for mixed-extensions of normal-form games. Proposition \ref{prop:reg-guarantee-gd} is then a restatement of:

\begin{theorem*}[B.1, \cite{ahunbay2024first}]
    Suppose in a smooth game that a subset of players $\hat{N}$ all employ online gradient ascent with the same step sizes $\eta_{t}$. Then for any Lipschitz-differentiable tangential function $h : \times_{i \in \hat{N}} X_i \rightarrow \mathbb{R}$, 
    $$ \sum_{t = 1}^T \sum_{i \in \hat{N}} \bilin{\nabla_i h(x_{\hat{N}}^t)}{\nabla_i u_i(x^t)} \leq \sum_{t = 1}^T \frac{h(x^{t+1}_{\hat{N}}) - h(x^{t}_{\hat{N}})}{\eta_t} + 2\eta_t L_h \left( \sum_{i \in \hat{N}} G_i \right)^2.$$
\end{theorem*}

\section{Connection to Continuous Time Dynamics}\label{sec:cont-dyn}

In a smooth game, given initial conditions $x_i(0) \in X_i$ for each player $i$, the gradient flow is the solution to the ``differential equation'',\footnote{We abuse terminology here in favour of intuition; for projected dynamics, the correct formalism for the gradient flow is that of a differential inclusion. However, whenever $\| \nabla_i u_i\|$ is bounded and $\nabla_i u_i$ is Lipschitz continuous for each player $i$, given any initial condition $(x_i)_{i \in N}$, $\exists$ unique solution $x: \mathbb{R}_+ \rightarrow X$ such that the equation holds almost-everywhere with respect to the Lebesgue measure on $\mathbb{R}$ \cite{nagurney2012projected}.}
$$ \frac{dx_i(\tau)}{d\tau} = \proj{\TC{X_i}{x_i(\tau)}}{\nabla_i u_i(x(\tau))} \ \forall \ \tau \geq 0.$$
\citet{ahunbay2024first} then observes that, at time $\tau$, the instantaneous rate of change of a tangential differentiable function $h : \times_{i \in \hat{N}} X_i \rightarrow \mathbb{R}$ is bounded below,
\begin{align}
\frac{dh(x(\tau))}{d\tau} & = \sum_{i \in \hat{N}} \bilin{\nabla_ih(x(\tau))}{\frac{dx_i(\tau)}{d\tau}} = \sum_{i \in \hat{N}} \bilin{\nabla_ih(x(\tau))}{\proj{\TC{X_i}{x_i(\tau)}}{\nabla_i u_i(x(\tau))}} \label{eq:cont-time-argument}\\
& \geq \sum_{i \in \hat{N}} \bilin{\nabla_ih(x(\tau))}{\nabla_i u_i(x(\tau))}. \nonumber \end{align}
Here, the last inequality follows since $\nabla_i h$ is tangent cone valued, whereas via Moreau's decomposition theorem, $\nabla_i u_i(x) = \proj{\TC{X_i}{x_i}}{\nabla_i u_i(x)} + \proj{\NC{X_i}{x_i}}{\nabla_i u_i(x)}$ for any $x \in X$. If $h$ is also bounded, then 
$$\int_0^T d\tau \cdot \sum_{i \in \hat{N}} \bilin{\nabla_ih(x(\tau))}{\nabla_i u_i(x(\tau))} \leq \int_0^T d\tau \cdot \frac{dh(x(\tau))}{d\tau} \leq h(T) - h(0),$$
which remains bounded even as $T \rightarrow \infty$. Proposition \ref{prop:reg-guarantee-gd} thus admits an interpretation that online gradient ascent well-approximates this property of the continuous-time projected gradient dynamics of the game, with the goodness of approximation determined by how small the step sizes are. 

\section{Non-Monotonicity of Natural Candidate Potential Functions in First-Price Auctions}\label{sec:potential-non-monotone}

To highlight the significance of our potential-function based proof of convergence, in this section we illustrate by example the failure of some classical and natural candidates for potential functions; they do not increase or decrease monotonically in the manner we expect in gradient dynamics. 
Hence, they cannot be used as the potential function in our proof.  These candidates include \emph{gradient norm} ($\ell_2$ and $\ell_1$), \emph{Nash gap}, and the \emph{expected second-bid}.

\begin{definition}
    In a first-price auction, at any mixed-strategy profile $x \in \times_{i \in N} \Delta(B_i)$, we define:
    \begin{enumerate}[label=(\alph*)]
        \item The $\ell_1$ and \textbf{square-$\ell_2$ gradient norm of buyer $1$} are respectively
        \begin{align*}
            \| \nabla_1u_1(x) \|_1 & = \sum_{b_1=1}^{v_1-1} \left| \sum_{b_{-1} \in B_{-1}} u_1(b)\left( \prod_{j \neq 1} x_j(b_j) \right) \right| = \ones^T \nabla_1u_1(x), \tag{via no-weak overbidding}\\
            \| \nabla_1u_1(x)\|^2 & =  \sum_{b_1=1}^{v_1-1} \left( \sum_{b_{-1} \in B_{-1}} u_1(b)\left( \prod_{j \neq 1} x_j(b_j) \right) \right)^2 = \nabla_1 u_1(x)^T \nabla_1u_1(x).
        \end{align*}
        \item The \textbf{Nash gap} is the sum of the maximum gains from deviations over all players,
        $$ \Delta u^*(x) = \sum_{i \in N} \max_{x'_i \in \Delta(B_i)} \{u_i(x'_i,x_{-i})\} - u_i(x).$$
        \item The \textbf{expected second-bid} is denoted, 
        $$ p_2(x) = \sum_{b \in B} \left( \prod_{j \in N} x_j(b_j)\right) \min_{i \in N} \max_{j \neq i} b_j.$$
    \end{enumerate}
\end{definition}

We remark that $\| \nabla_1u_1(x) \|_1,\| \nabla_1u_1(x)\|^2,\Delta u^*(x) \geq 0$ at every mixed-strategy profile $x$, and at first we would conjecture these quantities to be decreasing. In turn, the expected second-bid $p_2(x) \leq v_2-1$, and we would conjecture this quantity to be increasing. However, for both the continuous-time gradient dynamics (c.f. Appendix \ref{sec:cont-dyn}) and for projected gradient ascent, these conjectures turn out to be false even in first-price auctions with two buyers. Following notation in Appendix \ref{sec:cont-dyn}, we shall denote the rate of change of any function $f$ for the continuous-time dynamics of the game, at any mixed-strategy profile $x$ where $f$ is differentiable, as
$$ \frac{df(x)}{d\tau} = \sum_{i \in N} \bilin{\nabla_if(x)}{\proj{\TC{\Delta(B_i)}{x_i}}{\nabla_iu_i(x)}}.$$

\begin{example}\label{ex:1}
    Consider the setting with two buyers, with values $v_1 = 5$ and $v_2 = 3$, and bid sets $B_1 = \{0,1,2,3,4\}$, $B_2 = \{0,1,2\}$. Let $x_1 = (0.5,\ 0.1,\ 0.1,\ 0.15,\ 0.15)$ and $x_2 = (0.8,\ 0.1,\ 0.1)$. In this case, we have
    \begin{align*}
        \nabla_1u_1(x) & = \left(2,\frac{17}{5},\frac{57}{20},2,1\right), &
        \proj{\TC{\Delta(B_1)}{x_1}}{\nabla_1 u_1(x)} & = \left(-\frac{1}{4},\frac{23}{20},\frac{3}{5},-\frac{1}{4},-\frac{5}{4}\right),\\
        \nabla_2 u_2(x) & = \left( \frac{3}{4}, \frac{11}{10}, \frac{13}{20} \right), &
        \proj{\TC{\Delta(B_2)}{x_2}}{\nabla_2 u_2(x)} & =\left( -\frac{1}{12}, \frac{4}{15}, \frac{11}{60} \right).
    \end{align*}
    As a consequence, we calculate that 
    $$ \frac{d\|\nabla_1u_1(x)\|_1}{d\tau} = \frac{4}{15},~~  \frac{d\|\nabla_1 u_1(x)\|_2^2}{d\tau} = \frac{1490}{1423},~~  \frac{dp_2(x)}{d\tau} = -\frac{43}{600}. $$
    In particular, both $\ell_1$ and $\ell_2$ norms of $\nabla_1 u_1$ are strictly increasing at $x$, whereas the expected second-bid is strictly decreasing. Table \ref{table:1} shows that the directions of change for the quantities are the same after an iteration of projected gradient ascent with step size $\eta =0.01$.

    \begin{table}[H]
\centering
\caption{Non-monotonicity of the gradient norm and expected second-bid}
\label{table:1}
\begin{tabular}{|l|c|c|c|}
\hline
                                                                                             & $x$                 & $x_i \mapsto \proj{\Delta(B_i)}{x_i + \eta \nabla_i u_i(x)}$               & Change                              \\ \hline
\begin{tabular}[c]{@{}l@{}}\textsc{Gradient}\\ \textsc{Norm} ($\ell_1$)\end{tabular} & $11.25$                       &      $11.25267...$      &    $0.00267...$                            \\ \hline
\begin{tabular}[c]{@{}l@{}}\textsc{Gradient}\\ \textsc{Norm} ($\ell_2$)\end{tabular} & $28.68245...$                     &      $28.70344...$      &           $0.020997...$                     \\ \hline
\begin{tabular}[c]{@{}l@{}}\textsc{Expected}\\ \textsc{2nd Bid} \end{tabular}
& $0.14$                        &        $0.13930$     &    -0.007           \\ \hline
\end{tabular}
\end{table}
\end{example}

\begin{example}
    Again consider the setting of Example \ref{ex:1}, but for the mixed-strategy profile $x_1 = (0,\ 0.9,\ 0.1,\ 0,\ 0)$ and $x_2 = (0.4,\ 0.2,\ 0.4)$. This time, we have 
    \begin{align*}
        \nabla_1u_1(x) & = \left(1,2,\frac{12}{5},2,1\right), &
        \proj{\TC{\Delta(B_1)}{x_1}}{\nabla_1 u_1(x)} & = \left(0,-\frac{1}{5},\frac{1}{5},0,0\right),\\
        \nabla_2 u_2(x) & = \left( 0, \frac{9}{10}, \frac{19}{20} \right), &
        \proj{\TC{\Delta(B_2)}{x_2}}{\nabla_2 u_2(x)} & =\left( -\frac{37}{60}, \frac{17}{60}, \frac{1}{3} \right).
    \end{align*}
    As a consequence, we calculate that $d\Delta u^*(x)/d\tau = 293/600 > 0$, i.e. the Nash gap is strictly increasing at $x$. This is also true for projected gradient ascent; after a gradient update with step size $\eta = 0.01$, the Nash gap changes from $0.75$ to $\simeq 0.75586$, with an increase of $\simeq 0.00586$.
\end{example}

\section{Omitted Proofs}

\subsection{Order of Convergence in Proposition \ref{prop:reg-guarantee-gd}} \label{sec:extend-reg-gd}

First note that $\times_{i \in \hat{N}} X_i$ has bounded diameter equal to $2 \hat{N}$. Moreover, for any $x_{\hat{N}}, x'_{\hat{N}} \in \times_{i \in \hat{N}} X_i$,
$$ h(x_{\hat{N}}) - h(x_{\hat{N}}') \leq D_{\hat{N}} G_h.$$
By subtracting a constant if necessary, we may assume that $h$ takes values in $[0,D_{\hat{N}} G_h]$. Then, by rearrangement,
\begin{align*}
    \sum_{t = 1}^T \frac{h(x_{\hat{N}}^{t+1})-h(x_{\hat{N}}^t)}{\eta_t} = \frac{h(x^{T+1})}{\eta_T} - \frac{h(x_{\hat{N}}^1)}{\eta_1} + \sum_{t = 1}^{T-1} h(x^{t+1}) \left( \frac{1}{\eta_{t}} - \frac{1}{\eta_{t+1}} \right) \leq \frac{D_{\hat{N}}G_h}{\eta_T}.
\end{align*}
Here, since $(\eta_t)_{t \geq 1}$ is non-increasing by assumption, $(1/\eta_t)_{t \geq 1}$ is non-decreasing. Therefore, $\frac{1}{\eta_t} - \frac{1}{\eta_{t+1}} \leq 0$ for any $t$, from which the inequality follows. To conclude (\ref{eq:asymptotic-regret}), note that for any four $a,b,c,d \geq 0$, $ab + cd \leq (a+c)(b+d) = ab + cd + ad + cb$.

\subsection{Proof of Theorem \ref{thm:iterative}}\label{sec:proof-thm-iterative}

Formally, we begin by defining:

\begin{definition}
    For the mixed-extension of a normal-form game, the set of mixed-strategies $\times_{i \in N} \Delta(A_i)$ is a product of simplices. Then, for any collection $A'_i \subseteq A_i$ of subsets of actions for each player, we shall refer to $\times_{i \in N} \Delta(A'_i)$ as a \textbf{product of subsimplices}. Over any product of subsimplices defined via $(A'_i)_{i \in N}$, a \textbf{multilinear function} $\ell$ is of the form 
    $$\ell(x) = \sum_{a_i \in A'_i \ \forall \ i \in N} \left( \prod_{j \in N} x_j(a'_j) \right) c(a')$$ 
    for some $c : A' \rightarrow \mathbb{R}$. A function $f : X \rightarrow \mathbb{R}$ admits a \textbf{multilinear minorant} if there exists a multilinear function $\hat{f} : X \rightarrow \mathbb{R}$ such that $f(x) \geq \hat{f}(x)$ for every $x \in X$.
\end{definition}

We remark that the elimination condition (\ref{cond:lyapunov}) is equivalent to the condition that function $\sum_{i \in N} \bilin{\nabla_i h(\cdot)}{\nabla_i u_i(\cdot)} : \times_{i \in N} \Delta(A_i) \rightarrow \mathbb{R}$ admits a multilinear minorant of a particular form. However, in Algorithm \ref{alg:ievpp-h} (IEPP), for any iteration $\ell > 1$, $h^\ell$ admits such a multilinear minorant \emph{only on a product of subsimplices} $\times_{i \in N} \Delta(A_i^\ell)$. So we need to extend this multilinear minorant to all of $\times_{i \in N} \Delta(A_i)$, so that the associated time-average guarantees can be stitched together.

To this end, we first note the rate of change of $h$, $\sum_{i \in N} \bilin{\nabla_i h(\cdot)}{\nabla_i u_i(\cdot)}$, is Lipschitz continuous.

\begin{lemma}[\citep{ahunbay2024first}, in proof of Theorem 5.1]
    Suppose that $h : \times_{i \in N} \Delta(A_i) \rightarrow \mathbb{R}$ is Lipschitz-differentiable. Then its rate of change $\sum_{i \in N} \bilin{\nabla_i h(\cdot)}{\nabla_i u_i(\cdot)} : \times_{i \in N} \Delta(A_i) \rightarrow \mathbb{R}$ is Lipschitz continuous, with Lipschitz modulus $W^h \leq L_h \sum_{i \in N}G_i + G_h \sum_{i \in N} L_i$. 
\end{lemma}

As a consequence, in Algorithm \ref{alg:ievpp-h}, each $h^\ell$ admits a multilinear minorant on all of $\times_{i\in N} \Delta(A_i)$ satisfying certain desirable sign conditions.

\begin{lemma}\label{lem:key}
    Suppose that $h^\ell : \times_{i \in N} \Delta(A_i^\ell) \rightarrow \mathbb{R}$ in Algorithm \ref{alg:ievpp-h} extends to a Lipschitz-differentiable, tangential function $\hat{h}^\ell : \times_{i \in N} \Delta(A_i) \rightarrow \mathbb{R}$ satisfying $\partial \hat{h}(x) / \partial x_i(a_i) = 0$ for any $i \in N$, $a_i \in A_i \setminus A_i^\ell$. Then for any $x \in \times_{i \in N} \Delta(A_i)$,
    $$\sum_{i \in N}  \bilin{\nabla_i \hat{h}^\ell(x)}{\nabla_i u_i(x)} \geq \sum_{a \in A} \left( \prod_{j \in N} x_j(a_j)\right) \left[ \mathbb{I}[\exists i \in N, a_i \in A^{'\ell}_i] - \sum_{i \in N} 2W^{\hat{h}^\ell} \mathbb{I}[a_i \notin A_i^\ell]\right].$$
\end{lemma}

\begin{proof}
    Let $x \in \times_{i \in N} \Delta(A_i)$ be arbitrary. Since for any $i \in N$, $A_i^\ell \setminus A^{'\ell}_i \neq \emptyset$ in Algorithm \ref{alg:ievpp-h}, choose a $x^0 \in \times_{i \in N} \Delta(A_i)$ be such that for any player $i \in N$,
    \begin{enumerate}
        \item for any $a_i \in A_i \setminus A_i^\ell$, $x^0_i(a_i) = 0$,
        \item for any $a_i \in A_i^{'\ell}$, $x^0_i(a_i) = x_i(a_i)$, and 
        \item for any $a_i \in A_i^\ell \setminus A_i^{'\ell}$, $x^0_i(a_i) \geq x_i(a_i).$
    \end{enumerate}
    That is, $x^0$ corresponds to a mixed-strategy profile where each player $i$ only uses strategies $a_i \in A^\ell_i$. Moreover, the probability each such strategy $a_i$ is used is weakly higher than $x_i(a_i)$, with the inequality strict only if $a_i$ is not eliminated via a potential proof by $h^\ell$. 
    
    Now, construct a sequence $x^0, x^1, ...$ via the following procedure until termination; at step $k \geq 0$, if there exists a player $i$ and an action $a^k_i \in A_i \setminus A_i^\ell$ such that $x_i^k(a^k_i) < x_i(a^k_i)$, find $\delta^k : A_i^\ell \setminus A_i^{'\ell} \rightarrow \mathbb{R}$ such that for any $a_i \in A_i^\ell \setminus A_i^{'\ell}$, $\delta^k(a_i) \leq x_i^k(a_i) - x_i(a_i)$, and $\sum_{a_i \in A_i^\ell \setminus A_i^{'\ell}} \delta^k(a_i) = x_i(a^k_i)$. We then fix
    $$ x^{k+1}_j(a_j) = \begin{cases}
        x_i(a^k_i) & j = i \text{ and } a_i = a^k_i, \\
        x^k_i(a_i) - \delta^k(a_i) & j = i \text{ and } a_i \in A_i^\ell \setminus A_i^{'\ell}\text{, and} \\
        x^k_j(a_j) & \text{otherwise.}
    \end{cases}$$
    At each iteration $k \rightarrow k+1$, the number of player / strategy pairs $(i,a_i)$ such that $x^k_i(a_i) \neq x_i(a_i)$ decreases by at least $1$, and therefore this process terminates in finitely many steps with $x^{D+1} = x$. Moreover, for any $a_i \in A_i^{'\ell}$ and any $k \geq 0$, $x_i^k(a_i) = x_i(a_i)$. Then if at step $k$, player $i$ and action $a_i^k \in A_i \setminus A_i^\ell$ are chosen, then 
    \begin{align*}
        & \sum_{j \in N} \bilin{\nabla_j \hat{h}^\ell(x^{k+1})}{\nabla_j u_j(x^{k+1})} - \bilin{\nabla_j \hat{h}^\ell(x^{k})}{\nabla_j u_j(x^{k})} \\
        &  
        \geq  -W^{\hat{h}^\ell} \|x^{k+1} -x^k\| \\
        & \geq -W^{\hat{h}^\ell} \|x^{k+1} - x^k\|_1 \\ &=  -W^{\hat{h}^\ell} \sum_{j \in N} \sum_{a_j \in A_j} |x^{k+1}_j(a_j)-x_j^k(a_i) | \\ & =  -W^{\hat{h}^\ell} \cdot \left( (x_i(a^k_i)-0) + \sum_{a_i \in A_i^\ell \setminus A^{'\ell}_i} \delta^k(a_i) \right) \\ & = -2W^{\hat{h}^\ell} x_i(a_i^k).
    \end{align*}    
    Finally, at step $\ell$, Algorithm \ref{alg:ievpp-h} chooses $h^\ell$ which eliminates $(A^{'\ell}_i)_{i \in N}$. Since the support of $x^0$ satisfies $x^0(a_i) = 0$ for any $a_i \in A_i \setminus A_i^\ell$, $\partial h^\ell(x^0)/\partial x_i(a_i) = 0$ for any such pair $(i,a_i)$. Moreover, by the assumption that $h^\ell$ was chosen by Algorithm \ref{alg:ievpp-h}, since $x^0$ is an element\footnote{After throwing away the zero entries $x^0_i(a_i)$ for $a_i \in A_i \setminus A_i^\ell$.} of the product of subsimplices $\times_{i \in N} \Delta(A_i^\ell)$, by (\ref{cond:lyapunov}),
    $$ \sum_{i \in N} \sum_{a_i \in A_i^\ell} \frac{\partial h^\ell(x^0)}{\partial x_i(a_i)} \frac{\partial u_i(x^0)}{\partial x_i(a_i)} \geq \sum_{a \in A} \left( \prod_{j \in N} x_j(a_j) \right) \mathbb{I}[\exists i \in N, a_i \in A_i^{'\ell}].$$
    However, we remark that the condition $\partial\hat{h}^\ell(x)/\partial x_i(a_i) = 0$ for $i \in N, a_i \in A_i \setminus A_i^\ell$ ensures that
    $$ \frac{\partial\hat{h}^\ell(x^0)}{\partial x_i(a_i)} = \frac{\partial h^\ell(x_0)}{\partial x_i(a_i)} \text{ for any } a_i \in A_i^\ell.$$
    Therefore,
    \begin{align*}\sum_{i \in N}  \bilin{\nabla_i \hat{h}^\ell(x^0)}{\nabla_i u_i(x^0)} & = \sum_{i \in N} \left[ \sum_{a_i \in A_i^\ell} \frac{\partial h^\ell(x^0)}{\partial x_i(a_i)} \frac{\partial u_i(x^0)}{\partial x_i(a_i)} + \cancel{\sum_{a_i \in A_i \setminus A_i^\ell} \frac{\partial \hat{h}^\ell(x^0)}{\partial x_i(a_i)} \frac{\partial u_i(x^0)}{\partial x_i(a_i)}}^{}\right] \\ & \geq \sum_{a \in A} \left( \prod_{j \in N} x_j(a_j) \right) \mathbb{I}[\exists i \in N, a_i \in A_i^{'\ell}].\end{align*}
    Piecing all of this together,
    \begin{align*}
        & \sum_{i \in N} \bilin{\nabla_i \hat{h}^\ell(x)}{\nabla_i u_i(x)} \\ 
         = \ &  \sum_{i \in N} \bilin{\nabla_i \hat{h}^\ell(x^0)}{\nabla_i u_i(x^0)} + \sum_{k = 0}^D \sum_{i \in N} \left[ \bilin{\nabla_i \hat{h}^\ell(x^{k+1})}{\nabla_i u_i(x^{k+1})} - \bilin{\nabla_i \hat{h}^\ell(x^{k})}{\nabla_i u_i(x^{k})} \right] \\
         \geq \ & \sum_{i \in N} \bilin{\nabla_i \hat{h}^\ell(x^0)}{\nabla_i u_i(x^0)} - 2W^{\hat{h}^\ell} \sum_{i \in N} \sum_{a_i \in A_i \setminus A_i^\ell} x_i(a_i) \\
         \geq \ & \sum_{i \in N} \bilin{\nabla_i \hat{h}^\ell(x^0)}{\nabla_i u_i(x^0)} - 2W^{\hat{h}^\ell} \sum_{a \in A} \left(\prod_{j \in N} x_j(a_j)\right) \sum_{i \in N} \mathbb{I}[a_i \in A_i \setminus A_i^\ell] \\
         \geq \ & \sum_{a \in A} \left( \prod_{j \in N} x_j(a_j)\right) \left[ \mathbb{I}[\exists i \in N, a_i \in A^{'\ell}_i] - \sum_{i \in N} 2W^{\hat{h}^\ell} \mathbb{I}[a_i \notin A_i^\ell]\right].
    \end{align*}
    where the penultimate line follows from the fact that $\sum_{a_j \in A_j} x_j(a_j) = 1$ for any player $j$ and an exchange of sums.
\end{proof}

As promised, an inductive argument finalises our proof.

\begin{proof}[Proof (of Theorem \ref{thm:iterative})]
    Will show that for any $\ell \geq 0$, there exists positive $\epsilon^\ell_0,\epsilon^\ell_1,...,\epsilon^\ell_\ell$ such that 
    $$ \sum_{k = 0}^\ell \epsilon^\ell_k \bilin{\nabla_i \hat{h}^\ell(x)}{\nabla_i u_i(x)} \geq \sum_{a \in A} \left( \prod_{j \in N} x_j(a_j)\right)\mathbb{I}[\exists i \in N, a_i \in \cup_{k = 0}^\ell A^{'k}_i].$$
    Proceed by induction on $\ell$. The base case $\ell = 0$ is simply by choice of $h^\ell$ by Algorithm \ref{alg:ievpp-h}, so it remains to show the result for $\ell >0$ given it holds for $\ell-1$. By Lemma \ref{lem:key},
    \begin{equation}\label{eq:iter-proof-1} \sum_{i \in N} \bilin{\nabla_i \hat{h}^\ell(x)}{\nabla_i u_i(x)} \geq \sum_{a \in A} \left( \prod_{j \in N} x_j(a_j)\right) \left[ \mathbb{I}[\exists i \in N, a_i \in A^{'\ell}_i] - \sum_{i \in N} 2W^{\hat{h}^\ell} \mathbb{I}[a_i \notin A_i^\ell]\right].\end{equation}
    Meanwhile, by the induction hypothesis, there exists $\epsilon^{\ell-1}_1,\epsilon^{\ell-1}_2,..., \epsilon^{\ell-1}_{\ell-1}$ such that 
    \begin{equation}\label{eq:iter-proof-2}
        \sum_{k = 1}^{\ell-1} \sum_{i \in N} \epsilon^{\ell-1}_k\bilin{\nabla_i \hat{h}^k(x)}{\nabla_i u_i(x)} \geq \sum_{a \in A} \left( \prod_{j \in N} x_j(a_j)\right) \mathbb{I}[\exists i \in N, a_i \in \cup_{k = 0}^{\ell-1} A^{'\ell}_i].
    \end{equation}
    Denote $N^\ell = |\{i \in N \ | \ A_i^{'\ell}\neq\emptyset\}|$. Then noticing that for each player $i$, $A_i \setminus A_i^\ell = \cup_{k = 0}^{\ell-1} A_i^{'k}$, via evaluation of (\ref{eq:iter-proof-1}) $+ (2W^{\hat{h}^\ell}N^\ell+1)$(\ref{eq:iter-proof-2}), we see that we may set $\epsilon^\ell_\ell =1$, and $\epsilon^\ell_k = (2W^{\hat{h}^\ell}N^\ell+1)\epsilon^{\ell-1}_k$ for every $k < \ell$. Setting $\epsilon(\ell)  = \epsilon^{K-1}_\ell = \prod_{k > \ell} (2W^{\hat{h}^k}N^k+1)$ thus establishes the first part of Theorem \ref{thm:iterative}. Time-average guarantee then follows by applying Proposition \ref{prop:performance-guarantees}.
\end{proof}

\subsection{Proof of Lemma \ref{lem:deviation-form-cubic}}\label{sec:proof-deviation-form-cubic}

To show \ref{eq:deviation-form-cubic}, we note that 
\begin{align}
& \bilin{\nabla_i g^{S_i,c}_i(x_i)}{\nabla_i u_i(x)} \nonumber \\
= \ & \bilin{\nabla_i [h^{S_i}_i(x_i)(1-c^Tx_i)]}{\nabla_i u_i(x)} \nonumber \\
= \  & (1-c^Tx_i)\bilin{\nabla_i h_i^S(x)}{\nabla_iu_i(x)} - h_i^{S_i}(x_i) \bilin{c}{\nabla_i u_i(x)} \tag{chain rule \& bilinearity of inner product} \\
= \ & (1-c^Tx_i) \sum_{a_{-i} \in A_i} \sum_{a_i \in S_i} \left( \prod_{j \in N} x_j(a_j) \right) \left[ -u_i(a) + \frac{1}{|A_i\setminus S_i|} \sum_{a'_i \in A_i \setminus S_i} u_i(a'_i,a_{-i})\right] \tag{from \eqref{eq:deviation-form}} \\
& - h_i^{S_i}(x_i) \sum_{a \in A} \left( \prod_{j \neq i} x_j(a_j) \right) c(a_i) u_i(a). \label{eq:deviation-to-bound}
\end{align}
To get our final expression, we proceed by further manipulation on \eqref{eq:deviation-to-bound}:
\begin{align*}
    & h_i^{S_i}(x_i) \sum_{a \in A} \left( \prod_{j \neq i} x_j(a_j) \right) c(a_i) u_i(a) \\ = \ &  \sum_{a_{-i} \in A_{-i}} \sum_{a_i \in A_i \setminus S_i} \left( \prod_{j \neq i} x_j(a_j)\right) h_i^{S_i}(x_i)c(a_i) u_i(a) \tag{$a_i \in S_i \Rightarrow c(a_i) =0$} \\
     = \ & \sum_{a_{-i} \in A_{-i}} \sum_{a_i \in A_i \setminus S_i} \left( \prod_{j \neq i} x_j(a_j)\right) \sum_{a'_i \in S_i} x_i(a'_i) \frac{h_i^{S_i}(x_i)c(a_i)}{\sum_{a_i'' \in S_i} x_i(a''_i)}  \cdot u_i(a)\tag{$\sum_{a_i \in S_i}x_i(a_i) = \sum_{a''_i \in S_i} x_i(a''_i)$} \\
     = \ & \sum_{a_{-i} \in A_{-i}} \sum_{a'_i \in A_i \setminus S_i} \left( \prod_{j \neq i} x_j(a_j)\right) \sum_{a_i \in S_i} x_i(a_i) \frac{h_i^{S_i}(x_i)c(a'_i)}{\sum_{a_i'' \in S_i} x_i(a''_i)}  \cdot u_i(a'_i,a_{-i}) \tag{switch tags $a'_i \leftrightarrow a_i$} \\
     = \ & \sum_{a_{-i} \in A_{-i}} \sum_{a_i \in S_i} \left( \prod_{j \in N} x_j(a_j)\right) \frac{h_i^{S_i}(x_i)c(a'_i)}{\sum_{a_i'' \in S_i} x_i(a''_i)}  \cdot u_i(a'_i,a_{-i}). \tag{rearrangement}
\end{align*}

\subsection{Proof of Proposition \ref{prop:extension}}\label{sec:proof-cubic-extensible}

Note that for $T_i = S_i \cup (A_i \setminus A_i^\ell)$, we have $A_i \setminus T_i = A_i^\ell \setminus S_i$. Therefore, the function $g_i^{T_i,d}$ is given, for any $y_i \in \Delta(A_i)$
$$g_i^{T_i,d}(y_i) = \left( 1 - d^Ty_i \right) \left[ -\frac{1}{2|A_i^\ell \setminus S_i| }\left( 1 - \sum_{a_i \in A_i^\ell \setminus S_i} y_i(a_i) \right)^2 - \frac{1}{2} \sum_{a_i \in S_i} y_i(a_i)^2-\frac{1}{2} \sum_{a_i \in A_i \setminus A_i^\ell} y_i(a_i)^2\right].$$
To see that $g_i^{T_i,d}$ is an extension of $g_i^{S_i,c}$, for $x_i \in \Delta(A_i^\ell)$, let $y_i \in \Delta(A_i)$ be such that $y_i(a_i) = x_i(a_i)$ for any $a_i \in A_i^\ell$, and otherwise $y_i(a_i) = 0$. Then $g_i^{T_i,d}(y_i) = g_i^{S_i,c}(x_i)$. Tangentiality and Lipschitz-differentiability were proven in Theorem \ref{thm:cubic-candidates}. Moreover, whenever $y_i(a_i) = 0$ for any $a_i \in A_i \setminus A_i^\ell$, $\partial g_i^{T_i,d}(y_i)/\partial y_i(a_i) = 0$ as desired.

\subsection{Proof of Theorem \ref{thm:first-price}}\label{appendix:first-price}

By Proposition \ref{prop:first-elim} and Theorem \ref{thm:iterative}, in our model of the first-price auction, we may iteratively eliminate all bids of buyer $1$ which are $< v_2-2$ via an additively separable potential function whose value depends only the mixed-strategies of buyers $1$ and $2$. This covers Theorem \ref{thm:first-price}.(1), however, we still need to further eliminate strategies to show Theorem \ref{thm:first-price}.(2). The case $v_1 > v_2 + 2$ turns out to be straightforward, in fact following from the no-external regret property, whereas when $v_1 = v_2+2$ we will need to invoke the semicoarse equilibrium property.

In both cases considered, we shall be able to eliminate the bids solely through the semicoarse equilibrium property (Proposition \ref{prop:scce-constraints}). Then by \eqref{eq:deviation-form}, it will be sufficient to verify the Lyapunov condition \eqref{cond:lyapunov} only on pure bidding profiles. In all cases, restricting attention to the bidding strategies available to the buyers in the subgame, through \eqref{eq:good-devs}, we shall let $h_2 = h^{B_2\setminus \{v_2-1\}}_2$ be the function which generates the uniform deviation of buyer $2$ to bidding $v_2-1$.

\paragraph{Case (I): $v_1 > v_2 + 2$.} Here, we fix $h_1 = h_1^{B'_1 \setminus \{v_2\}}$, which generates the uniform deviation for buyer $1$ to the bid $b_1 = v_2$. In this case, for the Lyapunov condition \eqref{cond:lyapunov} to hold for $h = \alpha h_1 + \beta h_2$, it is sufficient to fix $\alpha, \beta \geq 0$ such that for any pure bidding profile profile $b$ in the subgame, 
\begin{equation}\label{eq:desired-c}
    \alpha( u_1(v_2,b_{-1})-u_1(b) ) + \beta (u_2(v_2-1,b_{-2}) - u_2(b)) \geq \mathbb{I}[b_1 \neq v_2].
\end{equation}
Proceed by case analysis on $b$. If $b_1 \geq v_2$, then by the no-weak overbidding condition on buyers $j \geq 2$, $u_2(v_2-1,b_{-2}) - u_2(b) = 0$ and
$$ u_1(v_2,b_{-1})-u_1(b) = (v_1-v_2) - (v_1-b_1) = b_1 - v_2 \geq \mathbb{I}[b_1 > v_2],$$
so $\alpha \geq 1, \beta \geq 0$ is sufficient. So suppose that $b_1 = v_2-1$. Denote by $\omega$ the number of buyers $j \geq 2$ who bid $v_2-1$. Then buyer $1$'s utility change is given, 
$$ u_1(v_2,b_{-1})-u_1(b) = (v_1-v_2) \left( 1 - \frac{1}{\omega+1}\right)-\frac{1}{\omega+1}.$$
In turn, buyer $2$'s utility change is 
$$ u_2(v_2-1,b_{-2}) - u_2(b) = \frac{\mathbb{I}[b_2<v_2-1]}{\omega+2}  \geq 0.$$
Now, if $\omega \geq 1$, then buyer $1$'s change in utility is at least $3-4/(\omega+1) \geq 1$ since $v_1-v_2 \geq 3$, and thus $\alpha \geq 1,\beta\geq 0$ remains sufficient. If $\omega =0$ then buyer $1$ has change in utility $-1$, but in this case buyer $2$ has change in utility $1/(\omega+2) = 1/2$. Therefore, $\alpha \geq 1$ and $\beta \geq 2(\alpha+1)$ is sufficient here.

The final case for buyer $1$'s bid is $b_1 = v_2-2$. If buyer $1$ is strictly outbid by some player $j \geq 2$, then $u_1(v_2,b_{-1})-u_1(b) = v_1 - v_2 \geq 3$, whereas $u_2(v_2-1,b_{-1})-u_2(b) \geq 0$ as buyer $2$ may only deviate out of losing bids. Thus $\alpha \geq 1, \beta \geq 2(\alpha+1)$ remain sufficient. So suppose that all players $j \geq 2$ bid $\leq v_2-2$, and similarly as before let $\omega$ be the number of buyers $j \geq 2$ who bid so. Then the change in utility of buyer $1$ is given, 
$$u_1(v_2,b_{-1})-u_1(b) = (v_1 - v_2) \left( 1 - \frac{1}{\omega+1} \right) - \frac{2}{\omega+1}.$$
Now, this quantity is increasing in $\omega$, and if $\omega = 1$, it equals $(v_1-v_2)/2-1 \geq 1/2$. Meanwhile, for any $\omega \geq 1$, the change in utility for buyer $2$ is $\geq 1-2/(\omega+1)$, which is non-negative for any $\omega \geq 1$. Finally, if $\omega = 0$, then buyer $1$ has change in utility $-2$. However, in this case, buyer $2$ must be bidding $< v_2-2$, and thus has utility change $1$. Thus, $\beta \geq 2\alpha+1$ is sufficient here. Since $N \geq 2$, we conclude $\alpha = 2$, $\beta = 2(\alpha+1) = 6$ ensure that \eqref{cond:lyapunov} holds for all pure bidding profiles in the subgame.

\paragraph{Case (II): $v_1 = v_2 + 2$.} This case proceeds by a two step strategy elimination; first we eliminate the bids $\{v_2-2,v_2+1\}$, after which buyer $1$ will be left with the set of bids $\{v_2-1,v_2\}$. Then in the second step, we shall eliminate the bid $v_2-1$. In each step of strategy elimination, we shall again look for a Lyapunov function $h$ of the form $\alpha h_1 + \beta h_2$.
    
Let $h_1 = h_1^{\{v_2-2,v_2+1\}}$ be the function which generates the strategy modification for buyer $1$, mapping bids $v_2-2,v_2+1$ to the uniform distribution on $\{v_2-1,v_2\}$. We shall demonstrate that, for any pure bidding profile $b$ in the subgame, the inequality
\begin{align}\label{eq:desired-a}
    & \alpha \cdot \mathbb{I}[b_1 \in \{v_2-2,v_2+1\}] \left(\frac{u_1(v_2-1,b_{-1})+u_1(v_2,b_{-1})}{2} - u_1(b)\right) \\ & + \beta \left( u_2(v_2-1,b_{-2}) - u_2(b) \right) \geq \mathbb{I}[b_1 \in \{v_2-1,v_2+1\}] \nonumber
\end{align}
holds, which implies the Lyapunov condition \eqref{cond:lyapunov}. Proceed by case analysis on $b$. If buyer $1$ bids $b_1 = v_1-1 = v_2+1$, then as buyers $j \geq 2$ may only place bids $\leq v_j-1 < v_2$, buyer $2$ has payoff $0$ no matter how $b_2$ is chosen. The no weak-overbidding condition also implies that $u_1(b) = 1$, $u_1(v_2,b_{-1}) = 2$ and $u_1(v_2-1,b_{-1}) \geq 3/N$. Therefore, in this case, $\eqref{eq:desired-a} \geq 3\alpha/2N$, and thus setting $\alpha \geq 2N/3$ is sufficient.

Meanwhile, if $b_1 \in \{v_2,v_2-1\}$, then the left-hand side of \eqref{eq:desired-a} is $\geq 0$; buyer $1$ does not deviate from their strategy, whereas buyer $2$ may only deviate out of guaranteed-to-lose bids $b_2 < v_2-1$. So consider the final subcase where $b_1 = v_2-2$. If there exists buyer $j \geq 2$ who bids $b_j = v_2-1$, then buyer $2$'s deviation is again weakly utility improving, whereas buyer $1$'s expected payoff improvement is lower bounded by $3/2N + 1$. Therefore, $\alpha \geq 2N/3$ remains sufficient in this case. 
    
If instead all buyers $j \geq 2$ bid $b_j \leq v_2-2$, denote by $\omega$ the number of such buyers $j \geq 2$ who bid exactly $v_2-2$. Then note that buyer $1$'s change in payoff becomes,
\begin{align*}
    \frac{u_1(v_2-1,b_{-1})+u_1(v_2,b_{-1})}{2}-u_1(b) = \frac{3+2}{2} - \frac{4}{\omega+1} = \frac{5}{2} - \frac{4}{\omega+1}.
\end{align*}
Now, if $\omega \geq 1$, then buyer $2$'s change in payoff satisfies 
$$ u_2(v_2-1,b_{-2}) - u_2(b) = 1 - u_2(b) \geq 1-\frac{2}{\omega+1},$$
where the lower bound is attained only if $b_2 = v_2-2$. This payoff change is $\geq 0$, so \eqref{eq:desired-a} holds whenever $\alpha \geq 2$ and $\beta \geq 0$. If instead $\omega =0$, then $b_2 < v_2 -2$ and thus $u_2(v_2-1,b_{-2}) - u_2(b) = 1$. In this case, $\beta \geq 1+3\alpha/2$ is sufficient. Therefore, we may fix $\alpha = \max\{2N/3,2\}$ and $\beta = 1 + 3\alpha/2$.

As a consequence, buyer $1$'s bids $\{v_2-2,v_2+1\}$ are both eliminated, and buyer $1$ is left with the bids $\{v_2-1,v_2\}$. To eliminate the bid $v_2-1$ via a potential proof by $h$, we let $h_1 = h_1^{\{v_2-1\}}$ be the uniform deviation to the bid $v_2$, and let $h = 2 h_1 + (4+N) h_2$, where the validity of the coefficients are established through analogous arguments.

\subsection{An Edge Case}\label{appendix:edge-case}

As we point out in Remark \ref{remark:edge-case}, the case when $v_1 = v_2+1$ turns out to be interesting. In this case, no matter the value of $v_3$, there is no equilibrium where buyer $1$ bids $v_2-2$. If $v_3 < v_2$, then there is an equilibrium where both buyer $1$ \& $2$ bid $v_2-1$. If instead $v_3 = v_2$, then in all equilibria of the game, buyer $1$ bids $v_2$. This can actually be proven through the semicoarse equilibrium property, where we continue on from Theorem \ref{thm:first-price} to eliminate the remaining bids.

\begin{proposition}\label{prop:edge}
    Suppose that $v_1 = v_2 +1$. If either $N = 2$ or $N > 2$ with $v_3 < v_2$, then buyer $1$'s bid $v_2-2$ may be eliminated via an additively separable potential of the form $h(x) = h_1(x_1) + h_2(x_2)$. If instead $v_2 = v_3$, then buyer $1$'s bids $\{v_2-2,v_2-1\}$ may be eliminated with additively separable potentials of the form $h(x) = h_1(x_1) + h_2(x_2) + h_3(x_3)$.
\end{proposition}

\begin{proof}
    All the potentials we shall consider will be those generating semicoarse equilibrium constraints, which implies as before that we may restrict attention to pure bid profiles to check the validity of the Lyapunov condition (\ref{cond:lyapunov}).

    In the first case, note that by the no-weak overbidding assumption, no buyer $j > 2$ may bid $v_2-1$. Let $h_1 = h_1^{\{v_2-2,v_2\}}$ prescribe to buyer $1$ the uniform deviation to $v_2-1$, whereas $h_2 = h_2^{\{0,1,..,v_2-1\}}$ will prescribe the same deviation to buyer $2$. Then if we let $h = 2 h_1 + 3 h_2$ works as a potential to eliminate the bid $v_2-2$; the key observation is that no buyer $j > 2$ may bid $v_2-1$. This implies that deviating from $b_1 = v_2$ does not hurt the payoff of buyer $1$, who goes from a guaranteed payoff of $1$ to a payoff of $2$ with probability at least $1/2$. Since we do not seek to eliminate the bid $v_2$, this is sufficient. Buyer $2$'s deviation payoffs are again necessarily non-negative, both in this case, and when $b_1 = v_2-1$ in which case buyer $1$ does not deviate. 
    
    In turn, if buyer $1$ bids $v_1-2$, we condition on the bid of buyer $2$. If $b_2 < v_2-2$, buyer $1$'s utility decreases from $3$ to $2$, but buyer $2$'s deviation provides a utility gain of $1$. If $b_2 = v_2-2$, then buyer $1$'s deviation provides a utility gain of $1/2$, and likewise with buyer $2$. If instead $b_2 = v_2-1$, then buyer $1$'s has a deviation payoff increase of $2$, whereas buyer $2$ does not deviate.

    In the second case, let $h_1 = h_1^{\{v_2-2,v_2-1\}}$ prescribe a uniform deviation to $v_2$ to buyer $1$, whereas for $i \in \{2,3\}$, we let $h_i = h_i^{\{0,1,...,v_2-2\}}$ prescribe uniform deviations to $v_2-1$ to both buyers. We will let $h = \alpha h_1 + \beta ( h_2 + h_3)$, leveraging the symmetry of buyers $2$ \& $3$. Then analogous case analysis shows that $\alpha = 3, \beta = 4$ is sufficient. 
    
    To wit, if $b_1 = v_2$, then no deviation for buyers $2$ \& $3$ result in a change in payoffs. If $b_1 = v_2-1$, then let $\omega = |\arg \max_{j \in N} b_j|$, and note that buyers $2$ \& $3$ may only deviate out of losing bids, resulting in non-negative utility changes. Buyer $1$ then incurs a change in payoff of $1-2/\omega$. If $\omega = 1$, then neither buyers $2,3$ are bidding $v_2-1$, so their deviation results in an increase of utility of $1/2$ each. Therefore, we require $-\alpha + 2\beta (1/2) \geq 1$, i.e. $\beta \geq \alpha +1$. If $\omega = 2$, then buyer $1$ incurs no change in utility, but at least one of buyers $2$ and $3$ has a bid which is strictly losing. This buyer's deviation increases their payoff by $1/3$, and thus we require $\beta \geq 3$. Likewise, if $\omega = 3$, then buyer $1$ has a payoff increase of $1/3$, so $\alpha \geq 3$ is sufficient.

    Finally, if $b_1 = v_2-2$, first that there exists $j > 1$ with $b_j = v_2-1$. In this case, all buyers deviate out of losing bids, and buyer $1$ has a utility gain of $1$. Thus we require $\alpha \geq 1$, and $\alpha \geq 3$ remains sufficient. So suppose that all buyers bid $\leq v_2-2$. Let $\omega = |\arg \max_{j \in N} b_j|$ again, and note that buyer $1$ has a payoff change of $1-3/\omega$. If $\omega =1$, then buyer $1$'s deviation causes a loss of $2$, but buyers $2,3$ both gain payoff $1$ from their deviation. Therefore, we require $-2\alpha + 2\beta \geq 1$, and $\beta \geq \alpha + 1$ remains sufficient. If $\omega = 2$ then buyer $1$ incurs a loss of $-1/2$, but in this case buyers $2$ \& $3$ have deviations which weakly increase their payoffs -- and at least one of them bids strictly less than $v_2-2$. This buyer has a strict payoff gain of $1$ from their deviation. Therefore, the necessary condition is $-\alpha/2 +\beta \geq 1$. Finally, if $\omega = 3$, then buyer $1$ has no change in utility, but buyers $2,3$ have present utility at most $2/3$. This means their deviation provides a gain of at least $1/3$, which implies that $\beta \geq 3$ is still sufficient.
\end{proof}

Proposition \ref{prop:edge} provides the following guarantees for time-average guarantee of online gradient ascent in first-price auctions, conditional on assumptions on buyer $3$. That we may let buyer $3$ run any no-external regret algorithm in Corollary \ref{cor:edge} follows since $\sum_{t = 1}^T \bilin{\nabla_3h_3(x^t_3)}{\nabla_3 u_3(x^t)}$ is the regret against deviating uniformly to the bid $v_2-1$ for buyer $3$.

\begin{corollary}\label{cor:edge}
    Suppose that buyers $1$ \& $2$ implement online gradient ascent for $T$ periods with non-increasing step sizes satisfying $1/\eta_{iT}, \sum_{t = 1}^T \eta_{iT} = o(T)$. If either $N = 2$ or $N > 2$ with $v_3 < v_2$, then in time average buyer $1$ chooses a bid $\notin \{v_2-1,v_2\}$ with probability $o(1)$. If instead $N > 2$ with $v_3 = v_2$, if buyer $3$ also incurs $o(T)$ regret against bidding $v_2-1$ uniformly, then in time average buyer $1$ chooses a bid $\neq v_2$ with probability $o(1)$.
\end{corollary}

What if there exists buyers $j > 2$ with $v_j = v_2$, who are allowed to behave arbitrarily? In this case, the semicoarse equilibrium property is insufficient simply do not work; from the proof of Proposition \ref{prop:edge}, we see that whether buyer $1$ wants deviate out of or into bidding $v_2-1$ or $v_2$ depends on the number of buyers who would bid $v_2-1$. In particular, if there are $> 1$ buyers who bid $v_2-1$, then deviating from $v_2$ to $v_2-1$ results in a strict utility loss. However, in the $v_2 > v_3$, we cannot use deviations into $\{v_2\}$ or $\{v_2-1,v_2\}$ to eliminate the bid $v_2-2$. Indeed, if buyer $1$ and buyer $2$ both bid $v_2-2$, then no deviation of buyer $2$ is utility improving. However, deviating from $v_2-2$ to $v_2$ results in a strict utility loss of $-1/2$, and deviating to the uniform distribution on $\{v_2-1,v_2\}$ results in a change in payoff of $(2+1)/2 - 3/2 = 0$. 

It turns out that we may prove nevertheless that in time average, buyer $1$ bids $v_2-2$ with probability $o(T)$ via a cubic potential. This allows us to make no assumption on the behaviour of buyers $j > 2$, however, the cost is a slower convergence guarantee.

\begin{proposition}\label{prop:adversarial}
    Suppose that buyers $1$ \& $2$ implement online gradient ascent for $T$ periods with non-increasing step sizes satisfying $1/\eta_{iT}, \sum_{t = 1}^T \eta_{iT} = o(T)$. Then in time average, buyer $1$ chooses the bid $v_2-2$ with probability $o(1)$.
\end{proposition}

\begin{proof}
    We will consider a potential of the form $h(x) = \alpha g_1(x_1) + \beta h_2(x_2)$, with $h_2 = h_2^{\{0,1,...,v_2-2\}}$ prescribing a uniform deviation to bidding $v_2-1$ to buyer $2$. Meanwhile, for buyer $1$, we will let $g_1 = g_1^{\{v_2-2\},c}$, where $c = (0,1/2,-1/2)$. Here, $\bilin{\nabla_2h_2(x_2)}{\nabla_2 u_2(x)}$ is multilinear in all $x_i$, whereas $\bilin{\nabla_1 g_1(x_1)}{\nabla_1 u_1(x)}$ is multilinear over all $x_i$ except for $x_1$. Therefore, in our analysis we may assume that all buyers $\neq 1$ are prescribed pure strategies $b_{-1}$ via $x_{-1}$, and proceed by case analysis.

    Note that, in the form of deviations, the change in utility $\bilin{\nabla_1 g_1(x_1)}{\nabla_1 u_1(x)}$ of buyer $1$ equals,
    \begin{align*}
        x_1(v_2-2) \left( \frac{u_1(v_2-1,b_{-1})+u_1(v_2,b_{-1})}{2} - u_1(v_2-2,b_{-1}) \right)(1-c^Tx_1) \\ + \frac{3}{8}x_1(v_2-2)^2\left( u_1(v_2-1,b_{-1})-u_1(v_2,b_{-1}) \right).
    \end{align*}
    This is because
    \begin{align*}
        h_1^{\{v_2-2\}}(x_1) & = -\frac{1}{4} \left(1-x_1(v_2-1)-x_1(v_2) \right)^2 - \frac{1}{2} x_1(v_2-2)^2 \\
        & = -\frac{1}{4}x_1(v_2-2)^2 - \frac{1}{2} x_1(v_2-2)^2 = -\frac{3}{4}x_1(v_2-2)^2.
    \end{align*}

    We proceed with our case analysis. First suppose that there exists some buyer $j > 1$ such that $b_j = v_2-1$. In this case, the utility gain of buyer $2$ is minimised when $b_2 = v_2-1$ also. Moreover, the change in the utility of buyer $1$ is bounded,
    \begin{align*}
        &\geq x_1(v_2-2)\cdot \left( \frac{2}{2N} + \frac{1}{2} - 0 \right) (1-c^Tx_1) + \frac{3}{8} x_1(v_2-2) \left(\frac{2}{2N} - \frac{1}{2} \right) \\
        & \geq \frac{x_1(v_2-2)}{4} - \frac{3x_1(v_2-2)^2}{16} \geq \frac{x_1(v_2-2)}{16}, 
    \end{align*}
    where we use the fact that $c^Tx_1 \in [-1/2,1/2]$ and $x_1(v_2-2)^2 \leq x_1(v_2-2)$. In particular, in this case, $\alpha \geq 16$ is sufficient.

    Second, suppose that for every $j > 1$, $b_j < v_2-2$. In this case, buyer $1$ deviates out of $v_2-2$, a strictly winning bid, and the change in payoff is bounded below by $-3 x_1(v_1-2)$, the contribution in expectation of buyer $1$ winning with a bid of $v_2-2$. However, note that buyer $2$ is posting a strictly losing bid, and deviating to a bid of $v_2-1$ provides a payoff $\geq x_1(v_1-2)$. Therefore, $\beta \geq 3\alpha + 1$ is sufficient here.

    Finally, suppose that all buyers $j > 2$ post bids $\leq v_2-2$, where for some buyer $j$, $b_j = v_2-2$. The payoff change for buyers $1$ and $2$ in this case are minimised when $b_2 = v_2-2$, and $b_j < v_2-2$ for every $j > 2$. In particular, the change in the utility of buyer $2$ equals
    $$ x_1(v_2-2) + \frac{x_1(v_2-1)}{2}  - x_1(v_2-2) \geq 0,$$
    but this quantity does not admit any lower bound in terms of $x_1(v_2-2)$. Meanwhile, the change in utility of buyer $1$ equals, 
    \begin{align*}
        & x_1(v_2-2) \left( \frac{2+1}{2} - \frac{3}{2}\right) (1-c^Tx_1) + \frac{3}{8}x_1(v_2-2)^2 (2-1) \\
        = \ & \frac{3}{8} x_1(v_2-2)^2.
    \end{align*}
    Therefore, we conclude that 
    $$\sum_{i =1}^2 \bilin{\nabla_i h(x)}{\nabla_iu_i(x)} \geq x_i(v_2-2)^2$$
    when we choose $\alpha = 16, \beta = 49$.

    It remains to convert this bound into a time average guarantee. Let $(h^\ell)_{\ell = 0}^{v_2-3}$ be the potentials generated as we iteratively eliminate the bids of buyer $1$ which are $< v_2-2$. Then the proof of Lemma \ref{lem:key} and Theorem \ref{thm:iterative} suggest that there exist strictly positive $(\epsilon(\ell))_{\ell = 0}^{v_2-3}$ such that for any\footnote{With eliminated strategies added back in.} mixed strategy profile $(x_i)_{i \in N}$ of the first price auction, 
    $$ \sum_{i = 1}^2\bilin{\nabla_i \left[ h+\sum_{\ell = 0}^{v_2-3} \epsilon(\ell)h^\ell \right](x)}{\nabla_i u_i(x)} \geq x_i(v_2-2)^2.$$
    In particular, after $T$ rounds of online gradient ascent, the time average of the \emph{square} probability of buyer $1$ making a bid of $v_2-2$ is bounded above by Proposition \ref{prop:performance-guarantees},
    $$ \frac{1}{T} \sum_{t = 1}^T x_1(v_2-2)^2 \leq \frac{C}{T} \sum_{i = 1}^2 \left( \frac{1}{\eta_{iT}} + \sum_{t = 1}^T \eta_{it} \right) \equiv \epsilon, $$
    where $C > 0$ is some constant. Inspecting the optimisation problem 
    $$ \max_{(x_1^t(v_2-2))_{t = 1}^T \geq 0} \sum_{t = 1}^T x_1^t(v_2-2) \textnormal{ subject to } \sum_{t = 1}^T x_1^t(v_2-2)^2 \leq \epsilon T,$$
    we see that the optimum is attained when $x_1^t(v_2-2) = \sqrt{\epsilon}$ for every $1 \leq t \leq T$. Therefore, 
    $$\frac{1}{T}\sum_{t = 1}^T x_1^t(v_2-2) \leq \sqrt{\frac{C}{T} \sum_{i = 1}^2 \left( \frac{1}{\eta_{iT}} + \sum_{t = 1}^T \eta_{it} \right)}.$$
\end{proof}

With the usual choice of step sizes $\eta_{it} \simeq 1/\sqrt{T}$ or $1/\sqrt{t}$ for both buyers $1$ and $2$, Proposition \ref{prop:adversarial} provides an upper bound of $O(T^{-1/4})$ on the probability that buyer $1$ makes a bid of $v_2-2$. Moreover, we remark that this slowdown in time average convergence also breaks the manner in which construct a Lyapunov function for the entire game, which required multilinear minorants throughout for the instantaneous rates of change. We see that the deeper study of potential construction with slower rates of convergence as an interesting future direction to explore.

\end{document}